\documentclass[preprint,10pt]{elsarticle}

\makeatletter
\def\ps@pprintTitle{%
\let\@oddhead\@empty
\let\@evenhead\@empty
\def\@oddfoot{\centerline{\thepage}}%
\let\@evenfoot\@oddfoot}
\makeatother

\journal{Physica D}

\usepackage{graphicx,bbm,setspace,amssymb,amsfonts,amsmath,mathtools,mathrsfs,stmaryrd,amsthm,bm,etoolbox,comment,hyperref,url,caption,subcaption,color,enumitem,algorithm,booktabs,microtype,nicefrac,lingmacros,nccmath,tree-dvips,tikz-cd}
\usepackage[noend]{algpseudocode}
\usepackage{algorithm}

\usepackage[colorinlistoftodos]{todonotes}
\usepackage[utf8]{inputenc} 
\usepackage[T1]{fontenc}    

\DeclareMathOperator*{\argmax}{argmax}
\usepackage[noend]{algpseudocode}

\newtheorem{theorem}{Theorem}[section]
\newtheorem{thm}[theorem]{Theorem}
\newtheorem{thm-defn}[theorem]{Theorem/Definition}
\newtheorem{lemma}[theorem]{Lemma}
\newtheorem{prop}[theorem]{Proposition}

\theoremstyle{definition}
\newtheorem{theorem1}{Theorem}[section]
\newtheorem{definition}[theorem1]{Definition}

\theoremstyle{remark}
\newtheorem{remark}[theorem]{Remark}

\newcommand{\ignore}[1]{}{}
\DeclareMathOperator{\magn}{mag}

\begin{document}

\begin{frontmatter}

\title{Equivalence relations and $L^p$ distances between time series with application to the Black Summer Australian bushfires}

\author[label1]{Nick James} 
\author[label2]{Max Menzies} \ead{max.menzies@alumni.harvard.edu}
\address[label1]{School of Mathematics and Statistics, University of Melbourne, Victoria, Australia}
\address[label2]{Beijing Institute of Mathematical Sciences and Applications, Tsinghua University, Beijing, China}

\begin{abstract}
This paper introduces a new framework of algebraic equivalence relations between time series and new distance metrics between them, then applies these to investigate the Australian ``Black Summer'' bushfire season of 2019-2020. First, we introduce a general framework for defining equivalence between time series, heuristically intended to be equivalent if they differ only up to noise. Our first specific implementation is based on using change point algorithms and comparing statistical quantities such as mean or variance in stationary segments. We thus derive the existence of such equivalence relations on the space of time series, such that the quotient spaces can be equipped with a metrizable topology. Next, we illustrate specifically how to define and compute such distances among a collection of time series and perform clustering and additional analysis thereon. Then, we apply these insights to analyze air quality data across New South Wales, Australia, during the 2019-2020 bushfires. There, we investigate structural similarity with respect to this data and identify locations that were impacted anonymously by the fires relative to their location. This may have implications regarding the appropriate management of resources to avoid gaps in the defense against future fires.

\end{abstract}

\begin{keyword}
Bushfires \sep Time series analysis \sep Metric spaces \sep Equivalence relations

\end{keyword}

\end{frontmatter}

\section{Introduction}
\label{sec:intro}

We begin with the following general question: can one define equivalence relations between time series, whereby two time series are declared equivalent if they differ only up to noise? This is our guiding heuristic; in this paper, we give one concrete way to understand this.

We begin with motivation: in set theory and throughout mathematics, an equivalence relation $\sim$ on a set $\mathcal{X}$ can simplify the set. Indeed, an equivalence relation partitions the original set into equivalence classes - subsets of $\mathcal{X}$ in which everything is mutually equivalent. By considering the set of equivalence classes $\mathcal{X}/\sim$, equivalent elements of the original $\mathcal{X}$ have been collapsed or divided down to the same element of $\mathcal{X}/\sim$; for this reason, the latter is called the quotient set.

If we wish to understand time series as equivalent if they differ only up to noise, and then wish to define distance measures that ``cancel out the noise'', we would like to define equivalence relations $\sim$ such that the quotient set $\mathcal{X}/\sim$ is \emph{metrizable}, namely can be imbued with an appropriate metric. Then, two equivalence classes of time series $A$ and $B$ would have an unambiguous distance $d$ apart. That is, any time series $X_t$ in the $A$ class and any time series $Y_t$ in the $B$ class would have a noise-invariant distance $d$ apart.

To concretely define these equivalence relations, we take the posture that time series consist of locally stationary segments separated by change points (or structural breaks) at which the statistical properties of the time series change, a frequent assumption in the statistical literature \cite{Dahlhaus1997}. If two time series have the exact same structural breaks and appropriately chosen statistical properties (up to completely identical distributions) in each of the segments, then the difference between them is random noise. This is of course just one way to build a framework envisioned by the paragraphs above.

 The following theorem proves that the above goals are possible. Formally, it combines Theorems \ref{thm:mean_theorem} and \ref{thm:var_theorem}:

\begin{thm}
\label{thm:main_theorem}
There exist equivalence relations $\sim$ on the space $\mathfrak{X}$ of time series, obtained from algorithmic determination of change points, such that the resulting quotient spaces $\mathfrak{X}/ \mathord{\sim}$ are metrizable. Further, such metrics are continuous with respect to perturbation of the time series. 
\end{thm}

The proof and construction of such equivalence relations build on several insights: the change point algorithms developed by \cite{Hawkins1977,Hawkins2003} and \cite{Ross2014}, the distances between change points of \cite{James2020_nsm}, and the construction of Lebesgue $L^p$ spaces as metrizable quotients by \cite{Riesz}.

The most significant equivalence relation between functions in measure theory is that of almost everywhere equivalence: define  $f \sim g$ if $f=g$ a.e., that is, $\{x:f(x)\neq g(x)\}$ is measure zero. We rely on this crucial equivalence relation for both the motivation and mathematics of our theorem and construction. 

We recall the construction of the $L^p$ spaces. Let $\mathcal{L}^p([0,H])$ be the vector space of all Lebesgue-measurable functions $f$ on $[0,H]$ such that $|f|^p$ has finite integral. We have a quotient mapping, defined and outlined in \cite{Billingsley},
\begin{align*}
  \mathcal{L}^p([0,H]) &\to \mathcal{L}^p([0,H])/ \mathord{\sim}=: L^p([0,H]),\\
  f &\mapsto [f].
\end{align*}
We will use this notation in what follows. One then equips this quotient space with the $L^p$ norms, as defined in \cite{Billingsley} or in Section \ref{sec:proofs} for our specific case.

\subsection{Overview of change point detection}

Many domains in the physical and social sciences are interested in the identification of change points/structural breaks in time series data. Introduced by \cite{Hawkins1977} and developed further by \cite{Hawkins2003,Hawkins2005}, change point algorithms are an evolving field of algorithms intended to determine such breaks at change points, when a change in the statistical properties of a time series has occurred.

In the statistical literature, focused on time series data, researchers have developed change point models driven by hypothesis tests, where $p$-values allow scientists to quantify the confidence in their algorithms \citep{Moreno2013, Bridges2015, Peel2015}. Change point algorithms generally fall within Bayesian or hypothesis testing frameworks. Bayesian change point algorithms \citep{Barry1993,Xuan2007,Adams2007,james2021_spectral} identify change point in a probabilistic manner and allow for subjectivity through the selection of prior distributions, but suffer from hyperparameter sensitivity and do not provide statistical error bounds or $p$-values, often leading to a lack of reliability.

Within hypothesis testing, \cite{RossCPM} outlines algorithmic developments in various change point models initially proposed by \cite{Hawkins1977}. Some of the more important developments in recent years include the work of \cite{Hawkins2003} and \citep{Ross2012,Ross2013,Ross2014}. Ross \cite{RossCPM} recently introduced the CPM package, which allows for flexible implementation of various change point models on time series data. Given the package's flexibility and efficient implementation, we build our methodology on this suite of algorithms.

There are also alternative approaches to change point analysis outside the statistical literature. For example, \cite{Id2005} assigns to a time series a change point score $z(t)$ that changes with time. Specifically, it measures the dissimilarity of the span of future sequential vectors of time series data with sequential vectors of past data, using matrix algebra.

\subsection{Overview of distance measures}

The application of metric spaces has provided the groundwork for research advancement in numerous fields. Within time series analysis, new applications of metrics between time series have been found by various authors \cite{Moeckel1997,Dose2005,Basalto2007,Basalto2008,Szkely2007,Mendes2018,Mendes2019}. Outside time series, the use of metrics has become a popular topic within the field of both statistics and machine learning, where a distance function is optimally tuned for a candidate task. There have been numerous applications, including computer vision \cite{Hua2007,Snavely2006}, text analysis \cite{Davis2008, Lebanon2006}, program analysis \cite{Ha2007} and many others \cite{Thorpe2017,Memoli2011,Arjovsky2017,Xing2002,Zhao2005}.

Numerous areas of research have found it fruitful to consider more general implementations of distance functions and measures, such as semi-metrics. Within time series analysis, A{\ss}falg et al. \cite{Afalg2007} use metrics restricted to particular time intervals to focus on discrepancy at key times, as may be informed by the field of application. The same authors develop a framework of ``threshold query execution'', where they restrict time series to unbroken intervals consistently above a certain cutoff threshold $\tau$, and may then measure distances based on the length and location of these threshold intervals \cite{Afalg2006,Afalg2006_2,Assfalg2006_3}, also proposing ways to best select $\tau$ \cite{Assfalg2006_4}. Batista et al. give an overview of measures between time series and various desirable properties of invariance they may have \cite{Batista2013}. In particular, they introduce a new complexity-adjusted distance between time series designed to satisfy a property of complexity-invariance. They take an elegant approach that simply multiplies the Euclidean metric between two time series by an ordered ratio of the time series' complexities, producing a semi-metric that satisfies a $\rho$-adjusted form of the triangle inequality. Too many authors to list have used dynamic time warping (DTW) between time series \cite{Keogh2003,Ding2008,Keogh2008_DTW,Alon2009_DTW,VanLaerhoven2009_DTW,Zhang2011_DTW}, with numerous authors proposing highly efficient means to speed up the algorithm \cite{SparseDTW,Rakthanmanon2012}.

Outside time series analysis, semi-metrics have proven particularly useful to measure discrepancy between sets of points, and have applications in image analysis \citep{Baddeley1992,Dubuisson1994,Gardner2014}, fuzzy sets \citep{Brass2002,Fujita2013,Gardner2014,Rosenfeld1985} and improved computation in such tasks \citep{Eiter1997,Atallah1983,Atallah1991,Shonkwiler1989}. An overview of such applications is given by \cite{Conci2017}.

In our previous work \cite{James2020_nsm}, we conduct a computational analysis of various existing metrics and semi-metrics between finite sets, both the traditional metrics such as Hausdorff and Wasserstein, and the more recently introduced semi-metrics used in applications above. We also develop a new semi-metric between finite sets, show some desirable properties compared to existing options, and apply this to measuring distance between time series' sets of structural breaks. Our methodology has two downsides: these semi-metrics do not obey the triangle inequality, and the distance computation does not record any of the time series' attributes between these structural breaks. It really reduces entirely from a distance between time series to a semi-metric between sets. Similar work has been undertaken among some of the aforementioned authors, but their approaches generally share the same limitations as our previous work. The threshold query execution of A{\ss}falg et al. \cite{Afalg2006,Afalg2006_2,Assfalg2006_3} may be used to create a distance between time series, but their proposed semi-metric does not use any of the data within the thresholded intervals. While they do not explicitly use it as a distance measure, the $z(t)$ score of \cite{Id2005} could be used to compare time series, but would also lose track of the values taken by the time series. And \cite{Afalg2007} simply uses pre-entered intervals to compare distance, which must be the same for each time series - it does not aim to detect or use structural breaks.

In this paper, we develop new procedures for measuring discrepancy between time series based on structural breaks, simultaneously ameliorating these downsides and proving Theorem \ref{thm:main_theorem} in the process. We associate to a time series a piecewise constant function that contains the data of the change points as well as the mean, variance, or other desired property, and embed these piecewise constant functions within $L^p([0,H])$. This procedure is highly flexible, building upon any available change point algorithm and recording any desired statistical property.

We also incorporate clustering and outlier detection, drawing upon numerous approaches to time series analysis that have been successfully applied in numerous disparate fields such as epidemiology \cite{Manchein2020,Li2021_Matjaz,james2022_covidinfectivity,Blasius2020,james2021_TVO,Perc2020,Machado2020}, environmental sciences \cite{james2022_CO2,Khan2020,james2021_hydrogen}, finance \cite{Drod2021_entropy,james_georg,Liu1997,Wtorek2021_entropy,james_arjun,Drod2001,james2022_stagflation,Gbarowski2019,james2021_MJW}, cryptocurrencies \cite{Sigaki2019,Drod2020_entropy,James2021_crypto2,Drod2020,Wtorek2020}, crime \cite{james2022_guns,Perc2013,james2023_terrorist}, the arts \cite{Sigaki2018_art,Perc2020_art}, and other fields \cite{Ribeiro2012,Merritt2014,james2021_olympics,Clauset2015}.

\subsection{Overview of the Australian bushfires and related environmental science}
\label{subsec:bushfire_intro}

The 2019-2020 bushfire season in New South Wales (NSW), Australia, attracted international attention. Multiple states of emergency were declared across the state \cite{fires_emergency,fires_emergency_2}, over five million hectares were burned \cite{Boer2020}, and the NSW Rural Fire Service stated ``2019/20 was the most devastating bush fire season in NSW history, truly unparalleled in more ways than one'' \cite{bushfire_bulletin}. Apart from the immediate damage, the bushfires subjected the entire state to poor, and often hazardous, air quality \cite{Ryan2021,Walter2020,Jalaludin2020,Vardoulakis2020}, which has been described as unprecedented by several researchers \cite{Johnston2020,BorchersArriagada2020}. Even after the pollution subsided, it may have caused significant damage to the ozone layer, causing further health risks \cite{Utembe2018,Bernath2022}.

As air quality data is a significant barometer of air pollution health risk and a proxy for the severity of bushfires nearby, we analyze air quality index (AQI) time series data from the New South Wales Department of Planning and Environment \cite{bushfireAQIdata}. We examine hourly data over a three-month window, October 20, 2019 - January 20, 2020, measured at 52 stations across the state, whose location is also recorded. This period represents the peak of the bushfire season, where poor AQI data was observed rather uniformly across the entire state. Our paper builds on a growing literature of using supervised and unsupervised learning methods to analyze air quality data, including random forest algorithms \cite{Grange2018,Grange2019}, generalized additive models \cite{Westmoreland2007}, principal component analysis \cite{Derwent1995}, and non-parametric smoothing \cite{Libiseller2005}. As bushfire events in Australia and elsewhere continue to increase in prevalence and intensity \cite{Sharples2016}, we hope such mathematical analysis of air quality and other data may continue to inform policymakers and the public of the precise dangers of such pollution, and allow informed decision-making and mitigation in advance.\\
 
This paper is structured as follows: in Section \ref{sec:proofs}, we construct the equivalence relations, associated quotient spaces, and metrics between time series, and prove all pertinent theorems. In Section \ref{sec:simulation}, we investigate desirable properties of our metrics relative to existing options, by both theoretical propositions and simulation. In Section \ref{sec:analysis_collections}, we describe our methodology for analysis of a collection of time series with our new metrics, including in a cross-contextual setting. In Section \ref{sec:bushfires}, we perform our analysis on the NSW air quality index time series. Section \ref{sec:conclusion} summarizes implications of our methodology and the findings regarding bushfires in Australia and elsewhere.

\section{Proof of theorems and construction}
\label{sec:proofs}

Let $(X_t)_{t \in I}$ be a real-valued time series over some time interval $I=\{0,1,...,H\}$, and $\mathfrak{X}$ the space of such time series. In this paper, we only record each data series once, so treat each $X_{t_i}$ as a single observed real number. Let $P_f \mathbb{R}$ be the set of finite subsets of $\mathbb{R}$. Throughout the paper, $p \geq 1$ is a chosen real number.

First, one selects a statistical attribute, such as mean, variance, or higher-order moment. For simplicity of exposition, we describe the procedure for the mean. We use the two-step algorithm of \cite{RossCPM} to determine a set of change points $S=\{c_1,...,c_m\}$ with respect to changes in the mean. This is an inductive deterministic procedure that uses iterated hypothesis testing. At each step, the null hypothesis of no change point existing is tested, and subject to preset parameters, the null hypothesis may be ruled invalid; this determines a change point. In addition, we set $c_0=0, c_{m+1}=H$. Properly interpreted, this gives a function
\begin{align}
CPM: \mathfrak{X} \to P_f \mathbb{R}.
\end{align}
Different parameters yield different algorithms and hence different such functions. This procedure has algorithmically determined that $(X_t)$ has a constant mean over each interval 
$[c_{i},c_{i+1}],i=0,...,m$. By using the change point algorithm or simply averaging $(X_t)$ over these intervals, we compute and record the empirically observed means $\mu_i$. Now we define a function 
\begin{align}
f= \sum_{i=0}^m \mu_i \mathbbm{1}_{(c_{i},c_{i+1})}
\end{align}
as a weighted indicator function. This is a piecewise constant function that codifies the change points of the time series and its changing means. That is,
\begin{align}
    f(x) =
    \begin{cases}
      \mu_{i} \text{ if } x \in (c_i,c_{i+1}),\\
      0 \text{ otherwise. }   \\
    \end{cases}
\end{align}
Let $\mathcal{F}(X_t):=[f] \in L^p([0,H])$ be the image of $f$ under the almost everywhere equivalence relation. This step means that the values of $f$ at the change points $c_i$ do not matter, for a finite set of points is measure zero. Let $V$ be the subspace of $L^p([0,H])$ consisting of (the images of) piecewise constant functions. Properly interpreted, our procedure is the function
\begin{align}
\mathcal{F}: \mathfrak{X} \to V.   
\end{align}
This defines an equivalence relation $\sim$ on $ \mathfrak{X} $ and an embedding 
\begin{align}
\mathfrak{X}/\mathord{\sim} \xhookrightarrow{} V \subset L^p([0,H]).
\end{align}
We note that $V$ is a vector subspace of $L^p([0,H])$ and thus inherits its $L^p$ norm $\|.\|_p$, defined by
\begin{align}
\|f\|_p= \left(\frac{1}{H} \int_0^H |f(x)|^p dx\right)^{\frac{1}{p}},
\end{align}
normalized so that $\|\mathbbm{1}_{[0,H]}\|=1$. When $p=2$, there is an inner product associated to the norm:
\begin{align}
\langle f,g\rangle =\frac{1}{H} \int_0^H f(x)g(x) dx.
\end{align}
At this point we must make some careful definitions and observations. While $\mathfrak{X}$ has a vector space structure, the change point algorithms do not respect addition of time series, and hence $\mathcal{F}$ is not linear. The absence of linearity forces us to differentiate two separate definitions. First, the metric space structure on $V$ induces a metric on $\mathfrak{X}/\mathord{\sim}$ and a pseudo-metric on $\mathfrak{X}$ itself. Concretely, if $(X_t),(Y_t)$ map to piecewise constant functions $f,g$ respectively, define $d_p(X_t,Y_t)=\|f-g\|_p$. 

Secondly and separately, we can pullback the norm structure to $\mathfrak{X}$ as a measure of the overall magnitude of the time series, defining $\magn(X_t)=\|\mathcal{F}(X_t)\|_p=\|f\|_p$. We use the symbol mag, rather than $\|.\|$, as this is neither a norm nor semi-norm on $\mathfrak{X}$. Concretely, $\magn(X_t)=d_p(X_t,0)$ but $d_p(X_t,Y_t)\neq \magn(X_t - Y_t).$ That is, mag does not induce the metric $d_p.$ Only after passage to $V$, via the non-linear $\mathcal{F}$, does the norm induce the metric.

Now we prove that $d_p$ has all properties of a metric on $\mathfrak{X}/\mathord{\sim}$.

\begin{enumerate}
    \item  It is immediate that $d_p(X_t,Y_t)\geq 0$ always.

\item It is immediate that $d_p(X_t,Y_t)=d_p(Y_t,X_t)$.

\item The triangle inequality $d_p(X_t,Z_t) \leq d_p(X_t,Y_t) + d_p(Y_t,Z_t)$ follows from Lemma \ref{triangle inequality lemma}.

\item If $d_p(X_t,Y_t)=0$, then $(X_t),(Y_t)$ are equal elements of $\mathfrak{X}/\mathord{\sim}$ by Lemma \ref{equivalence relation lemma}.

\end{enumerate}

\begin{lemma}
\label{triangle inequality lemma}
If $(X_t),(Y_t),(Z_t) \mapsto f,g,h \in V$  then $\|f-h\|_p \leq \|f-g\|_p + \|g-h\|_p$, establishing the triangle inequality for $d_p$.
\end{lemma}
\begin{proof}
Passing to $V$ it suffices to prove for any elements of $V$, $\|f+g\|_p \leq \|f\|_p + \|g\|_p.$ This is known as Minkowski's inequality, \cite{Minkowski}. First, the concavity of $\log$ implies Young's inequality: if $a,b\geq 0, p,q>1, \frac{1}{p}+\frac{1}{q}=1$ then
\begin{align}
ab \leq \frac{a^p}{p} + \frac{b^q}{q}.
\end{align}
Next, one proves H\"older's inequality:
\begin{align}
\|fg\|_1 \leq \|f\|_p \|g\|_q.
\end{align}
To prove this, without loss of generality, we can normalize $f$ and $g$ so that both $\|f\|_p=\|g\|_q =1.$ It remains to prove $\|fg\|_1 \leq 1.$ Young's inequality gives, for $x\in[0,H]$,
\begin{align}
|f(x)g(x)| \leq \frac{|f(x)|^p}{p}+\frac{|g(x)|^q}{q}.
\end{align}
We integrate this to get $|fg\|_1\leq 1$ as required. Finally, 
\begin{align}
\|f+g\|^p_p=&\frac{1}{H}\int^H_0 |f+g|^p dx\\
&\leq \frac{1}{H}\int^H_0 (|f|+|g|)|f+g|^{p-1} dx\\
&=\|f|f+g|^{p-1}\|_1 + \|g|f+g|^{p-1}\|_1\\
&\leq \|f\|_p \|(f+g)^{p-1}\|_q + \|g\|_p \|(f+g)^{p-1}\|_q\\
&= (\|f\|_p + \|g\|_p)\|f+g\|^{p-1}_p.
\end{align}
The first inequality is by the triangle inequality, the second is by H\"older. The final equality holds as $q(p-1)=p.$ We divide by $\|f+g\|^{p-1}_p$, assuming it is non-zero, to get the result. If $\|f+g\|^{p-1}_p$ is zero, the result is trivial.
\end{proof}
\begin{lemma}
\label{equivalence relation lemma}
Given time series $(X_t),(Y_t)$, the following are equivalent:

\begin{enumerate}
    \item  $d_p(X_t,Y_t)=0$;

 \item $X_t \sim Y_t$;

\item $X_t,Y_t$ have the same set of change points $c_0,...,c_m$ and the same values of means $\mu_i$ between them.

\end{enumerate}

\end{lemma}
In particular, equivalence does not depend on the parameter $p$. Figure \ref{fig:PeturbationPlot} depicts some examples of equivalent time series.

\begin{proof}
Let $(X_t)$ and $(Y_t)$ map to $f$ and $g$ respectively. Condition 3. states that $f=g$. By definition of $\sim$, condition 2. states that $[f]=[g] \in V$, while condition 1. states that $\|f-g\|_p=0$. Clearly then, 3. implies 2., and 1. and 2. are equivalent since $\|.\|_p$ is a norm on $L^p([0,H]).$ Finally, we suppose 1. holds. Then $f-g$ is itself a piecewise linear function with $L^p$ norm zero. If $f$ and $g$ have breaks $c_1,...,c_m$ and $d_1,...,d_n$, respectively, then $f-g$ is continuous on $[0,H] -\{c_1,...,c_m,d_1,...,d_n\}$. Since the integral of $|f-g|^p$ is zero over this region, continuity means $f-g$ is identically zero on $[0,H] -\{c_1,...,c_m,d_1,...,d_n\}$. As such, all the open intervals $(c_i,c_{i+1})$ and $(d_i,d_{i+1})$ must coincide, with the exact same $\mu_i$ coefficients. This proves condition 3. 
\end{proof}
\begin{lemma}
\label{perturbation property lemma}
Let $X_t,Y_t$ be time series, and suppose $X_t$ is compared with a perturbation $X'_t.$ Suppose $X'_t$ differs from $X_t$ on an interval $[t_0,t_0 + \delta]$ such that the change point algorithm $\mathcal{F}$ detects two additional change points at $t_0, t_0 + \delta$ and a mean that differs by $\epsilon$ from the mean of $X_t$. Then
\begin{align}
        \max \{ |d_p(X'_t,Y_t) - d_p(X_t,Y_t)|, |\magn(X_t)-\magn(X'_t)| \} 
        \leq \epsilon\left(\frac{ \delta}{H}\right)^{\frac{1}{p}}.
\end{align}
That is, $d_p$ and $\magn$ are continuous with respect to small deformations in the time series. Figure \ref{fig:PeturbationPlot} illustrates the suppositions of the lemma.
\end{lemma}
\begin{proof}
If $X_t,X'_t$ and $Y_t$ map to $f,f'$ and $g$, respectively then
\begin{align}
\max(|d_p(X'_t,Y_t) - d_p(X_t,Y_t)|,|\magn(X_t) - \magn(X'_t)|)\\
=\max(\|f'-g\|_p - \|f-g\|_p, \|f\|_p - \|f'\|_p) \\
\leq \|f-f'\|_p 
\leq \left(\frac{1}{H} \int^{t_0+\delta}_{t_0} \epsilon^p   \right)^{\frac{1}{p}}
\leq \epsilon\left(\frac{ \delta}{H}\right)^{\frac{1}{p}}.
\end{align}
We remark that if $p < \infty$ then this establishes continuity relative to the length of the deformation interval, in addition to the deformation magnitude.
\end{proof}
Combining all these results we have proved the following:

\begin{figure}
    \centering
    \includegraphics[width=\textwidth]{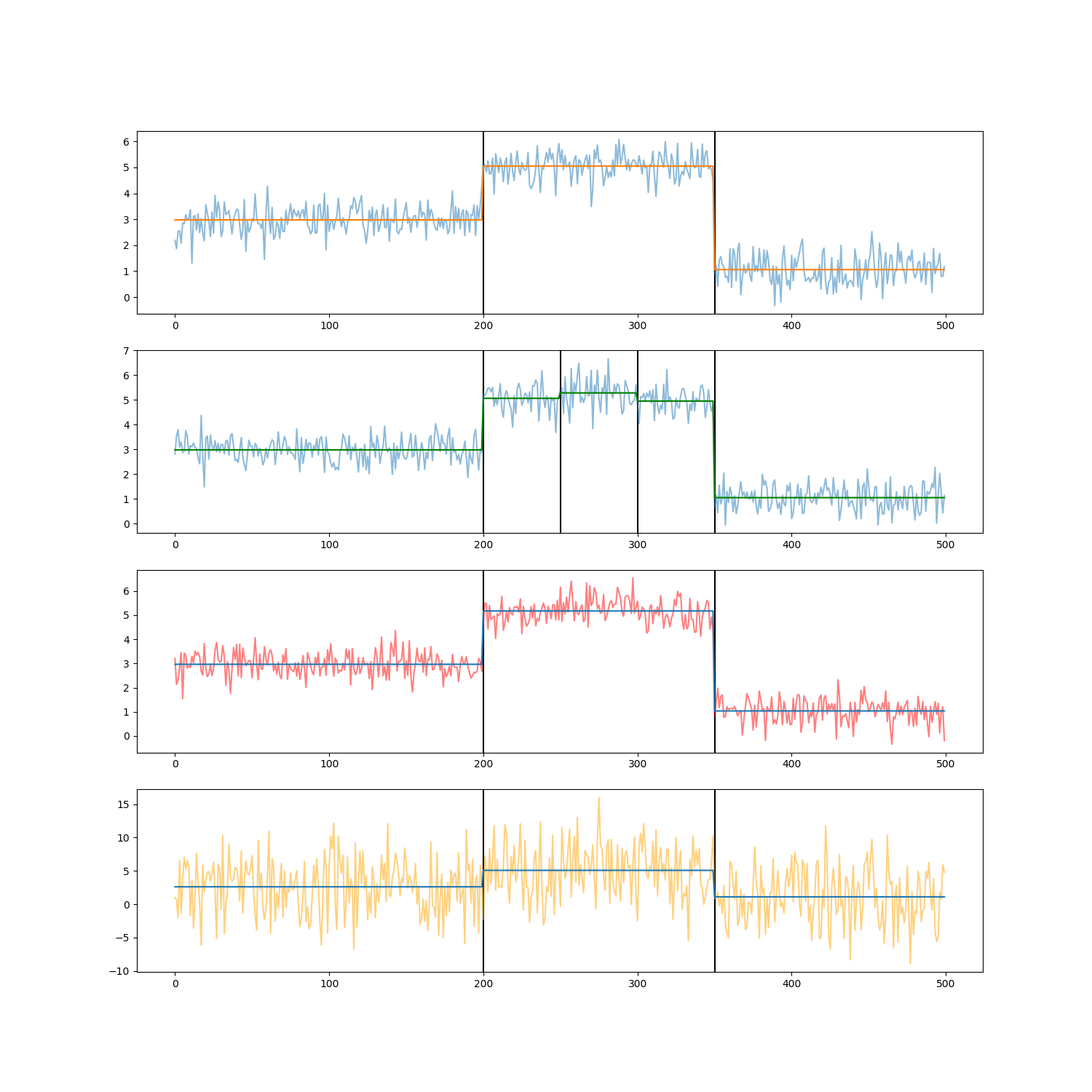}
    \caption{Illustration of equivalence and perturbation. Plot 1 depicts a time series $X_t$, Plot 2 depicts a time series $X'_t$ that differs from $X_t$ by a small deformation, as in Lemma \ref{perturbation property lemma} and proposition \ref{other perturbation prop}. Plots 1 and 3 are equivalent under both $\mathcal{F}_{\mu}$ and $\mathcal{F}_{var}$, as defined in Theorems \ref{thm:mean_theorem} and \ref{thm:var_theorem}. Plots 1 and 4 are equivalent under $\mathcal{F}_{\mu}$ but not $\mathcal{F}_{var}$. These plots illustrate the aim of our framework, to identify two time series as ``the same'' if they only differ up to noise. The precise definition may be complex, with Plots 1 and 3 more strongly equivalent - the same distributions in their locally stationary segments - than Plots 1 and 4 - just the same means in their locally stationary segments.
}
    \label{fig:PeturbationPlot}
\end{figure}

\begin{theorem}[Mean Theorem] 
\label{thm:mean_theorem}
The space $\mathfrak{X}$ is equipped with a class of functions $\mathcal{F}_{\mu}$ that induce metrizable quotient spaces $(\mathfrak{X}/\mathord{\sim},d_p)$. The equivalence can be characterized simply as $X_t \sim Y_t$ if $X_t$ and $Y_t$ have identical sets of change points and means in their stationary periods, and does not depend on $p$. The metrics $d_p$ are stable under perturbations in the time series. We may also equip $\mathfrak{X}$ with a magnitude function, representing a holistic measure of its changing mean up to noise.
\end{theorem}

By an identical method, using change point algorithms for changes in variance, one can construct deterministic functions $\mathcal{F}_{var}$ that prove the following:

\begin{theorem}[Variance Theorem]
\label{thm:var_theorem}
$\mathfrak{X}$ is equipped with a class of functions $\mathcal{F}_{var}$ that induce metrizable quotient spaces $(\mathfrak{X}/\mathord{\sim},d_p)$. The equivalence can be characterized as $X_t \sim Y_t$ if $X_t$ and $Y_t$ have identical sets of change points and variances in their stationary periods, and does not depend on $p$. The metrics $d_p$ are stable under perturbations in the time series. We may equip $\mathfrak{X}$ with a magnitude function, representing a holistic measure of its changing variance up to noise.

\end{theorem}

Other statistical properties induce other equivalence relations and maps $\mathcal{F}$, all producing variants of Theorem \ref{thm:main_theorem}. There exist change point algorithms that detect changes in all stochastic properties simultaneously; these induce finer equivalence relations.

\section{Simulation study and comparison of distances}
\label{sec:simulation}
In this section, we investigate some desirable properties of our metrics $d_p$ relative to existing options, by both theoretical propositions and simulation. First, we briefly compare our new distance to the simple Euclidean distance between time series.

\begin{prop}
\label{prop:noiseinvariance}
Let $\mathcal{F}$ be a mapping induced by a chosen change point algorithm procedure as in Section \ref{sec:proofs}. Let $X_t \sim X'_t$ and $Y_t \sim Y'_t$ be two pairs of equivalent time series, that is, $d_p(X_t,X'_t)=d_p(Y_t,Y'_t)=0$. Then $d_p(X_t,Y_t)=d_p(X'_t,Y'_t)$. However, regarding the Euclidean metric, the values of $\|X_t - Y_t\|_2$ and $\|X'_t - Y'_t\|_2$ may differ by a linearly increasing function of  $\| X_t - X'_t\|_2$ and $\| Y_t - Y'_t\|_2$.
\end{prop}
Heuristically, $X_t,X'_t$ and $Y_t, Y'_t$ differ only up to noise. Our distance $d_p$ filters out the noise and returns the same distance between $X_t$ and $Y_t$ as $X'_t$ and $Y'_t$. However, if the noise between $X_t$ and $X'_t$ is ``opposite'' that of the noise between $Y_t$ and $Y'_t$, the Euclidean metric may be sensitive to this and add these errors rather than remove them.
\begin{proof}
By definition, $X_t$ and $X'_t$ are mapped under $\mathcal{F}$ to a common element $[f] \in V$, while $Y_t$ and $Y'_t$ are mapped to a common element $[g]$. By definition, $d_p(X_t,Y_t)=\|f - g\|_p = d_p(X'_t,Y'_t)$, concluding the first statement. Now we turn to the Euclidean metric. By the triangle inequality, 
\begin{align}
\|X'_t - Y'_t\|_2 = \|X_t - Y_t + X'_t - X_t + Y_t - Y'_t\|_2 \\
\leq \|X_t - Y_t\|_2 + \|X'_t - X_t\|_2 + \|Y_t - Y'_t\|_2,\\
\text{so } \|X'_t - Y'_t\|_2 - \|X_t - Y_t\|_2 \leq \|X'_t - X_t\|_2 + \|Y_t - Y'_t\|_2.
\end{align}
To show this bound can be attained, we suppose $X'_t = X_t + N_t$ while $Y'_t = Y_t - N_t$ where $N_t$ is ``noise'' such that $X'_t$ and $X_t$ are deemed identical under $\mathcal{F}$ and analogously for $Y'_t$ and $Y_t$. Then $X'_t - Y'_t = X_t - Y_t + 2 N_t$. For simplicity, assume under the Euclidean inner product that $\langle N_t, X_t - Y_t \rangle = 0$, which would be the case if $X_t-Y_t$ were locally constant and $N_t$ had mean zero over the locally constant intervals of $X_t - Y_t$. Then by the Pythagorean theorem, 
\begin{align}
\|X'_t - Y'_t\|^2_2 = \|X_t - Y_t\|^2_2 + 4 \|N_t\|^2, \\
\text{so } \|X'_t - Y'_t\|_2 - \|X_t - Y_t\|_2 = \frac{4 \|N_t\|^2}{ \|X'_t - Y'_t\|_2 + \|X_t - Y_t\|_2} \\
\geq \frac{4 \|N_t\|^2}{ 2 \|X_t - Y_t\|_2 +  \|X'_t - X_t\|_2 + \|Y_t - Y'_t\|_2} \\
=  \frac{4 \|N_t\|^2}{2 \|X_t - Y_t\|_2 + 2 \|N_t\|},
\end{align}
which may increase linearly with $\|N_t\|$ when it is a linear proportion of $\|X_t - Y_t\|$, completing the proof.
\end{proof}

Next, we examine the performance of our metrics $d_p$ relative to previously proposed such distances between time series based on structural breaks. Our previous work \cite{James2020_nsm} uses metrics and semi-metrics between the sets of change points to measure distance between time series; that work does so without using data of the mean or variance between these breaks. The following simulation study and proposition provide justification for why the distance $d_p$ may more appropriately measure distance between time series.

For this simulation study, we will make use of the following (semi)-metrics between finite sets $S$ and $T$. First, let 
\begin{align}
d(x,S) = \inf_{s \in S} d(x,s). \label{min distance defn}
\end{align}
Then, let the Hausdorff metric be defined as
\begin{eqnarray}
\label{eq:Hausdorff}
    d_{H}(S,T) & = & \text{max } \bigg( \sup_{s \in S} d(s,T), \sup_{t \in T} d(t,S) \bigg), \\
    & = & \sup \{ d(s,T), s \in S; d(t,S), t \in T \}. 
\end{eqnarray}
The modified Hausdorff semi-metric is defined as
\begin{align}
\label{eq:modifiedHausdorff}
d^1_{MH}(S,T) = \max \bigg( \frac{1}{|S|} \sum_{s \in S} d(s,T), \frac{1}{|T|} \sum_{t \in T} d(t,S) \bigg).
\end{align}
The new family of semi-metrics introduced in \cite{James2020_nsm} is defined by 
\begin{align}
\label{eq:MJ}
    d^p_{MJ}({S},{T}) = \Bigg(\frac{\sum_{t\in T} d(t,S)^p}{2|T|} + \frac{\sum_{{s} \in {S}} d(s,T)^p}{2|S|} \Bigg)^{\frac{1}{p}}.
\end{align}
These aforementioned (semi-)metrics between finite sets can readily be applied to measure distance between time series with respect to their change points. If $(X_t)$ and $(Y_t)$ are time series with change point sets $S,T$, respectively (under some change point algorithm), one can simply set $d(X_t,Y_t)$ to be any of $d_H(S,T), d^1_{MH}(S,T)$ or $d^1_{MJ}(S,T).$

\begin{figure}
    \centering
    \includegraphics[width=\textwidth]{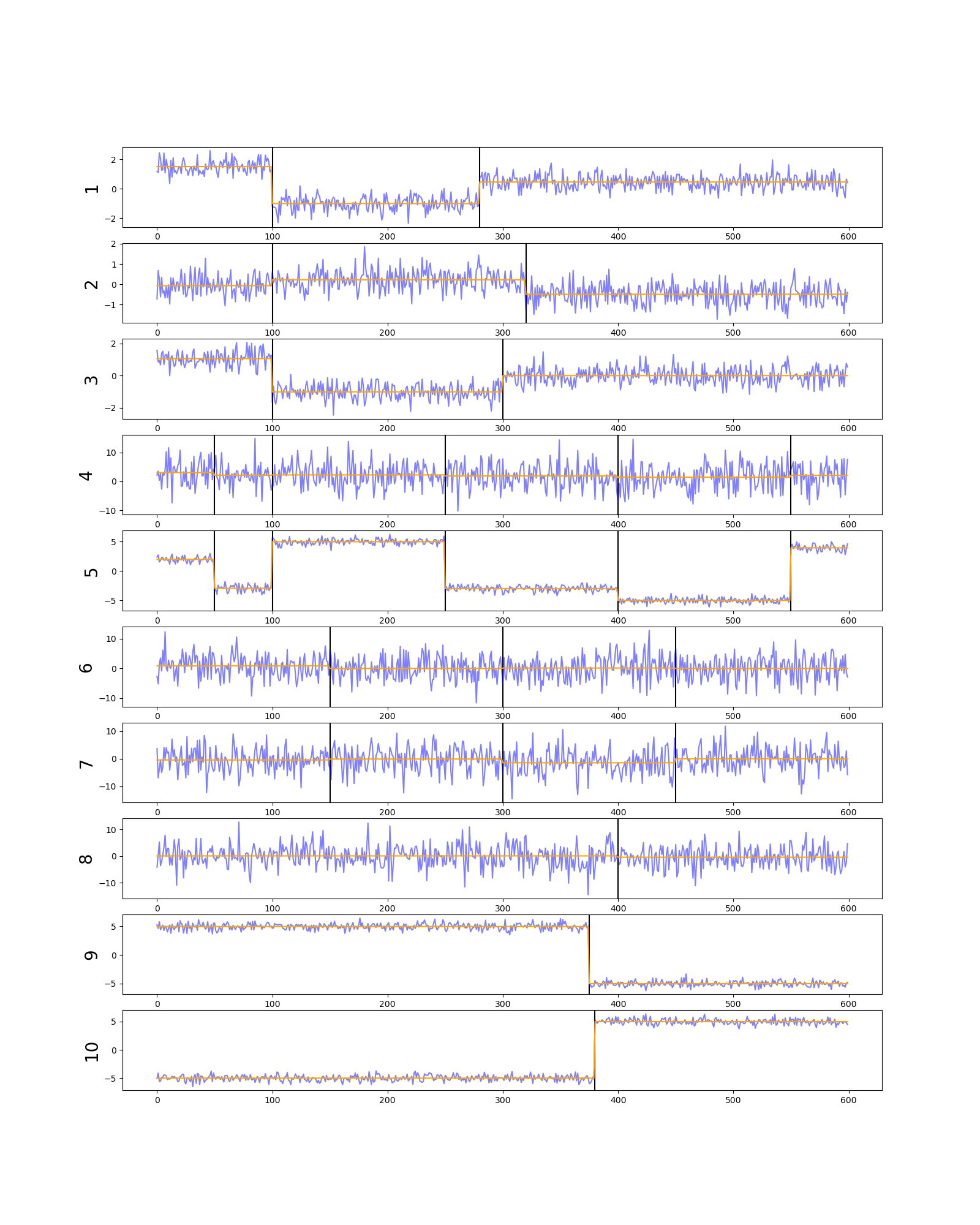}   \caption{Ten synthetic time series with synthetic change points and piecewise constant mean functions. These time series are locally stationary with respect to their mean, and give synthetically generated examples similar to the sorts of real-data time series to which one would apply our framework.
    }
    \label{fig:SyntheticData}
\end{figure}

\begin{figure}
    \centering
    \begin{subfigure}[b]{0.49\textwidth}
        \includegraphics[width=\textwidth]{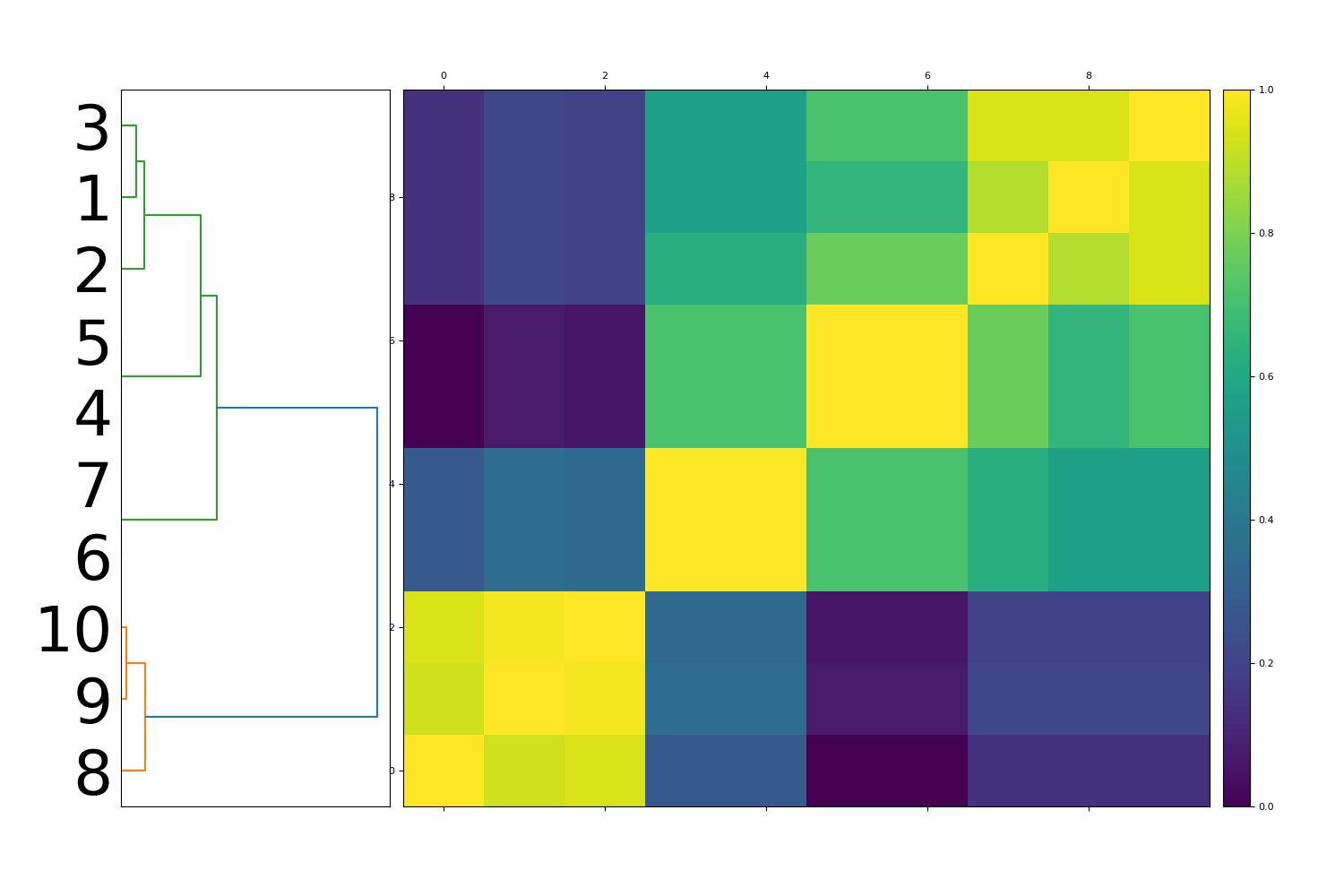}
        \caption{}
        \label{fig:Hausdorff_dendrogram}
    \end{subfigure}
    \begin{subfigure}[b]{0.49\textwidth}
        \includegraphics[width=\textwidth]{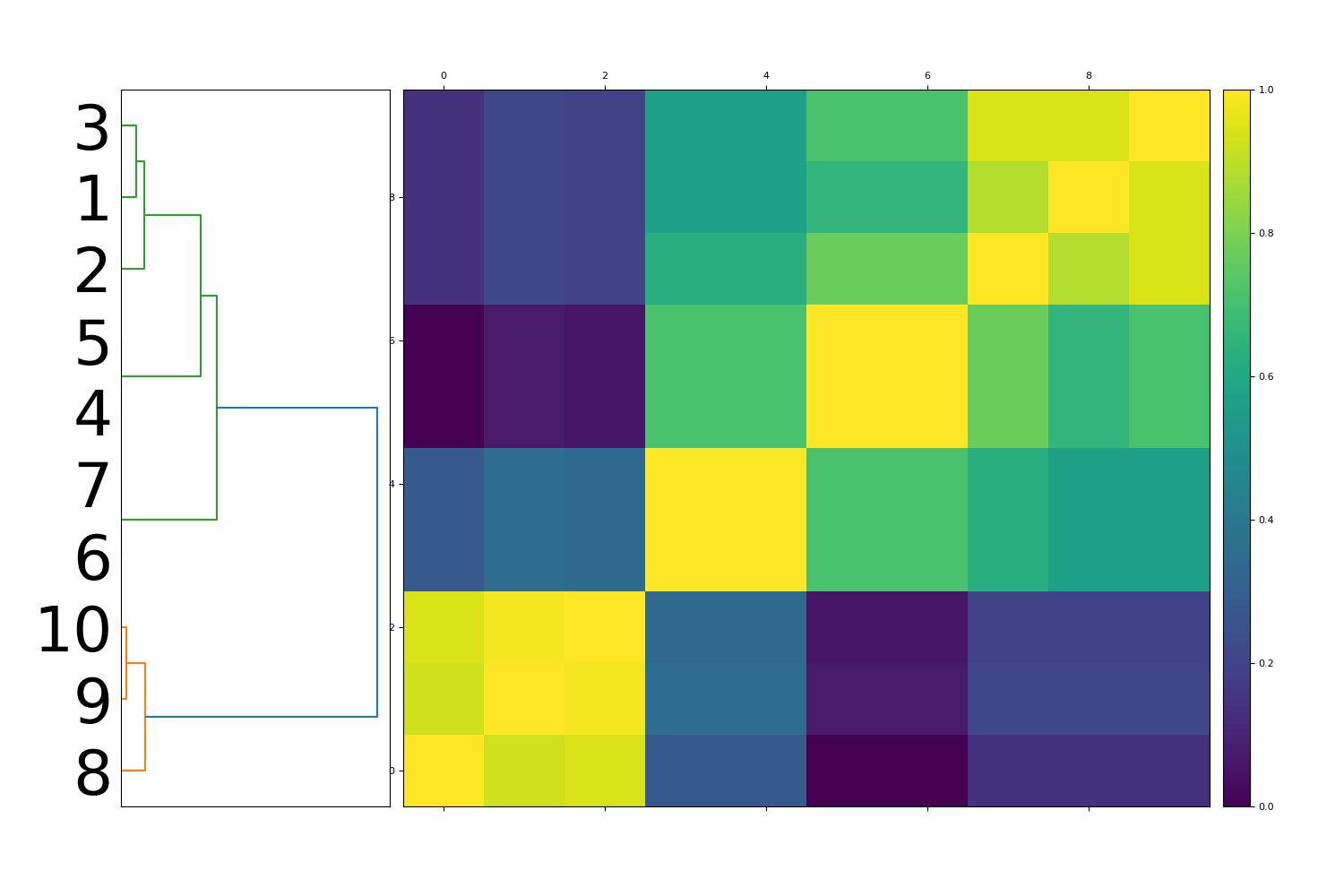}
        \caption{}
        \label{fig:Modified_Hausdorff_dendrogram}
    \end{subfigure}
    \begin{subfigure}[b]{0.49\textwidth}
        \includegraphics[width=\textwidth]{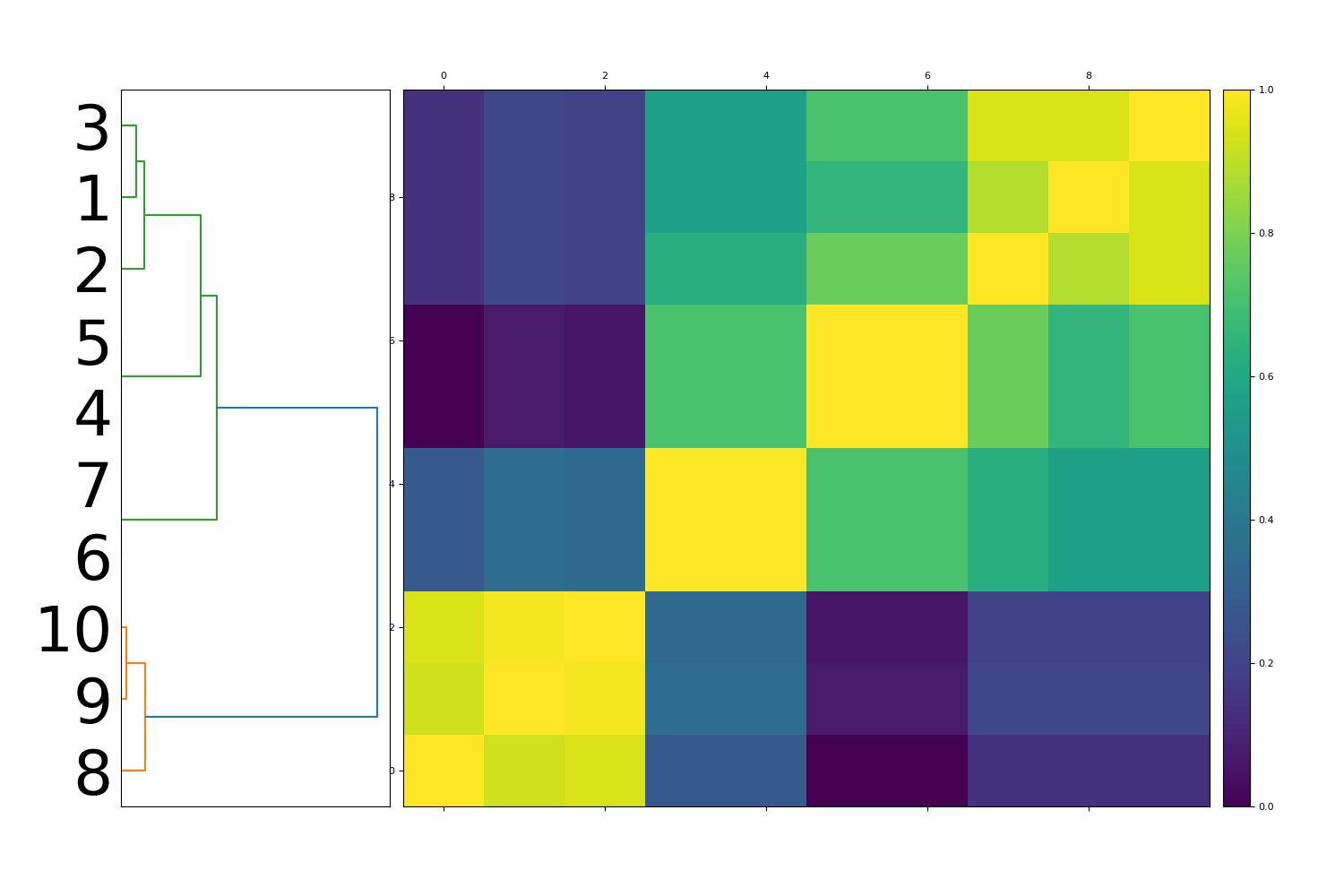}
        \caption{}
        \label{fig:MJ1_dendrogram}
    \end{subfigure}
    \begin{subfigure}[b]{0.49\textwidth}
        \includegraphics[width=\textwidth]{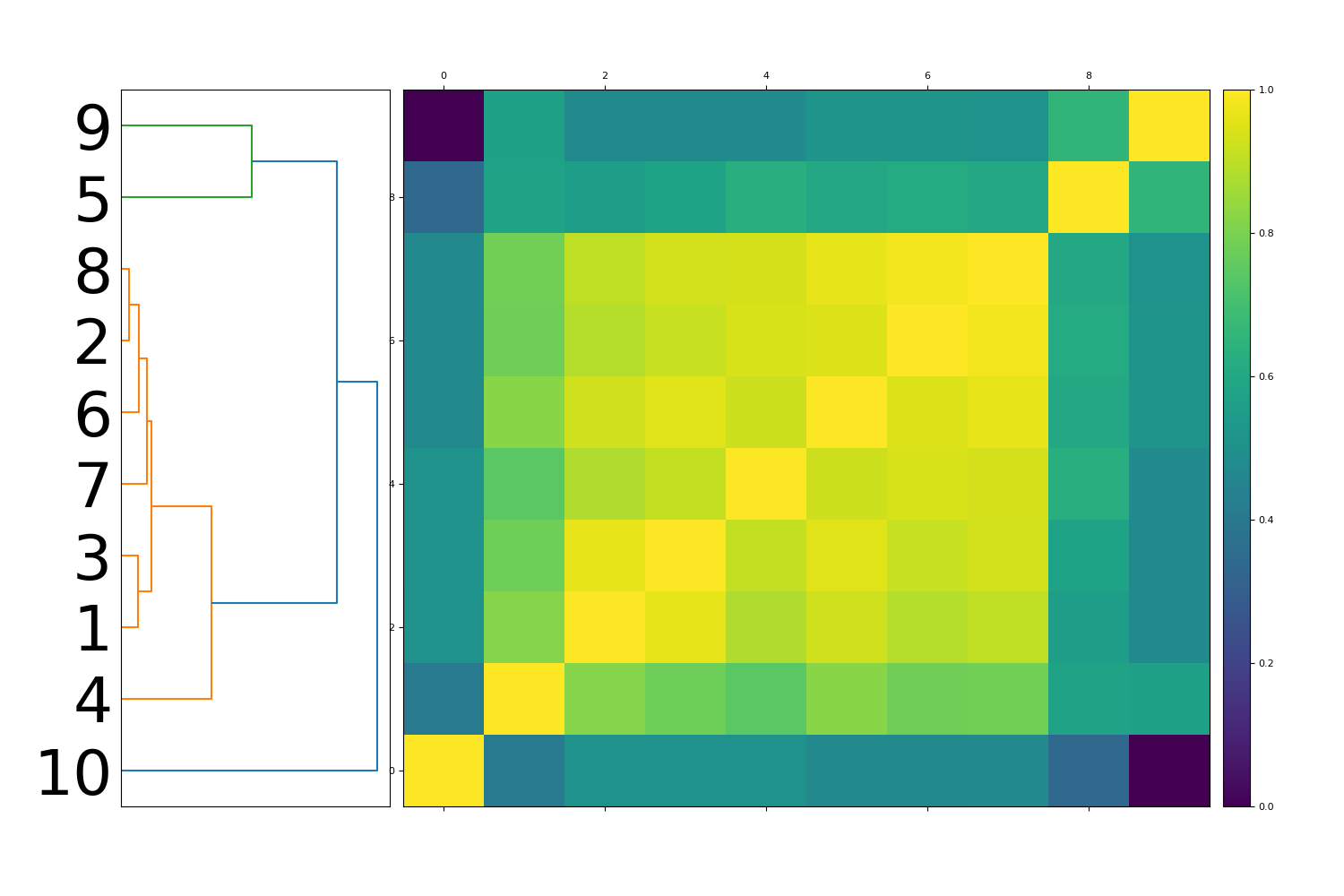}
        \caption{}
        \label{fig:Lp_dend}
    \end{subfigure}    
\caption{Hierarchical clustering applied to distances between ten synthetic time series using one of four discrepancy measures, (a) the Hausdorff metric between change points (\ref{eq:Hausdorff}), (b) the modified Hausdorff semi-metric between change points (\ref{eq:modifiedHausdorff}), (c) MJ$_p$ semi-metric between change points (\ref{eq:MJ}) and (d) $d_p$ distances between time series, with $p=1$. The final distance (d) is able to cluster the time series more appropriately without the loss of information inherent in (a), (b), (c); those distances identify erroneous similarities between the time series solely according to their structural breaks, whereas the $d_p$ distance recognizes more data.
}
\label{fig:Synthetic Distance inference}
\end{figure}

We generate ten synthetic time series (labeled TS1 - TS10), which we display in Figure \ref{fig:SyntheticData}. These provide example where two different structural breaks (change points) may have markedly different changes in the mean. Black vertical lines indicate detected change points, with the orange line representing the piecewise constant functions $f$. When considering only structural breaks, the most similar time series are groups as follows: \{TS1, TS2, TS3\}, \{TS4, TS5\}, \{TS6, TS7\}, \{TS8, TS9, TS10\}. However, considering structural breaks and their mean, using our metric $d_p$ induced from $\mathcal{F}_{\mu}$, the most similar time series are \{TS1, TS3\} and \{TS7, TS8\}.

In Figure \ref{fig:Synthetic Distance inference}, we implement hierarchical clustering on the ten synthetic time series with our new $d_p$ metric, here with $p=1$, as well as existing distance measures in \cite{James2020_nsm}. We see that only the $d_p$ metric, in Figure \ref{fig:Lp_dend}, appropriately clusters similar time series according to their changing values, while the other three figures identify erroneous similarities between the time series based solely on their structural breaks. That is, the other (semi-)metrics neglect the additional data that $d_p$ recognizes, demonstrating the superior performance of the latter.

Next, we prove a proposition that demonstrates that only our new metrics $d_p$ are continuous with respect to small deformations in the time series that may introduce new structural breaks. That is, the existing distance measures between change points may be excessively sensitive to new structural breaks, either caused by a deformation in the time series data or a faulty detection by a change point algorithm. We adopt the same suppositions of Lemma \ref{perturbation property lemma}, illustrated in Figure \ref{fig:PeturbationPlot}.

\begin{prop}
\label{other perturbation prop}
Let $X_t,Y_t$ be time series, and suppose $X_t$ is compared with a perturbation $X'_t.$ Suppose $X'_t$ differs from $X_t$ on an interval $[t_0,t_0 + \delta]$ such that the change point algorithm $\mathcal{F}$ detects two additional change points at $t_0, t_0 + \delta$ and a mean that differs by $\epsilon$ from the mean of $X_t$. The three (semi)-metrics $d_H, d^1_{MH},d^1_{MJ}$ between sets are not continuous with respect to deformations in the time series. 
\end{prop}

\begin{proof}
Adopting the notation of Lemma \ref{perturbation property lemma}, let $X_t,Y_t$ have structural breaks $S=\{c_1,...,c_m\}$ and $T=\{d_1,...,d_n\}$ respectively. Relative to $X_t$, $X'_t$ has two additional structural breaks $S'=S\cup\{t_0,t_0+\delta\}$. If we measure distance between time series purely relative to structural breaks using the Hausdorff metric $d_H$, then the distance between $X_t$ and $Y_t$ is defined as $d_H(S,T)$ while that between $X'_t$ and $Y_t$ is $d_H(S',T)$. As such, $|d_H(S,T)-d_H(S',T)|$ is constant and non-zero relative to $\epsilon$. Hence, it does not approach zero as $\epsilon \to 0$. The same holds for $d^1_{MH}$ and $d^1_{MJ}$. 

Moreover, these metrics are not continuous relative to $\delta$ either. Consider the simple case where $S=T$. Assume $c_{j-1}<t_0<t_0+\delta<c_j$. Then $d_H(S',T) \geq d(t,S) =\min(t_0 - c_{j-1}, c_j - t_0)$. This is also constant relative to $\delta$ and does not approach zero as $\delta \to 0$. A similar argument holds for $d^1_{MH},d^1_{MJ}$. This proves the result.
\end{proof}

Next, we prove a property that establishes a form of linearity for the metrics $d_p$ analogous to the Euclidean metric or other classical $L^p$ metrics.

\begin{prop}
\label{prop:linearity}
Let $X_t,Y_t$ be time series and $a,b$ be constants. Suppose the mapping $\mathcal{F}$ has one of two properties: either
\begin{align}
\label{eq:propcond1}
\mathcal{F}(aX_t+b)=a\mathcal{F}(X_t)+b \in V, \\
 \text{or } \mathcal{F}(aX_t+b)=|a|\mathcal{F}(X_t) \in V. \label{eq:propcond2}
\end{align}
Then $d_p(aX_t+b,aY_t+b)=|a|d_p(X_t,Y_t)$. That is, $d_p$ respects linear operations between time series.
\end{prop}
The proof of this proposition is quite trivial, but understanding the primary condition is more interesting. The property that $\mathcal{F}(aX_t+b) = a\mathcal{F}(X_t)+b \in V$ really has two components: first, that the change points of $aX_t+b$ are the same as $X_t$, assuming $a \neq 0$, and secondly, that the statistical quantity $\mathcal{F}$ is measuring scales commutes with scalar multiplication and the addition of constants. The property (\ref{eq:propcond2}) has a near-identical interpretation. The former component is more subtle than it looks, relying on the fact that the change point algorithm use normalized test statistics. Fortunately, it is the case that essentially all change point algorithms do use normalized test statistics (and should!) so that means this condition holds for all procedures we will use. As for the latter component, the property (\ref{eq:propcond1}) holds for the mean and hence $\mathcal{F}_\mu$ of Theorem \ref{thm:mean_theorem}. This is simply due to the fact that the mean of a random variable $Z$ satisfies $E(aZ+b)=aE(Z)+b$. On the other hand, the property (\ref{eq:propcond2}) holds for the standard deviation and hence $\mathcal{F}_\sigma$ closely related to $\mathcal{F}_{var}$ of Theorem \ref{thm:var_theorem}. This is also simply due to the fact that $\sigma(aZ+b)=|a|\sigma(Z)$.
\begin{proof}
In this proof, we assume the condition (\ref{eq:propcond1}) but the exact same proof holds for (\ref{eq:propcond2}). Suppose $X_t$ is mapped to $[f] \in V$ while $Y_t$ is mapped to $[g]$. Thus $d_p(X_t,Y_t)=\| f-g\|_p$ by definition. By the property of $\mathcal{F}$, we have $\mathcal{F}(aX_t+b)=af+b$, so $aX_t+b$ is mapped to $af+g$. Thus, $d_p(aX_t+b,aY_t+b)=\|af+b - (ag+b)\|_p = |a| \|f-g\|_p = |a|d_p(X_t,Y_t)$.
\end{proof}
The explanation above and the simple proof show that this proposition applies to the mean-mapping $\mathcal{F}_\mu$ and associated distances of Theorem \ref{thm:mean_theorem} as well the standard deviation-mapping $\mathcal{F}_{\sigma}$, related to the variance-mapping $\mathcal{F}_{var}$ of Theorem \ref{thm:var_theorem}. That is, $d_p$ respect linear operations between time series (scalar multiplication and the addition of a constant) for both mean and standard deviation.

We conclude this section by considering several invariances our $d_p$ metrics do and do not have, and some modifications one may make for a desired application. First, Proposition \ref{prop:noiseinvariance} immediately establishes noise-invariance, when properly defined, the guiding property of the entire framework. Next, as a result of Proposition \ref{prop:linearity}, $d_p$ are not invariant to amplitude (caling $X_t \mapsto aX_t$) or offset in the $y$-direction (when $X_t$ is left invariant but $Y_t \mapsto Y_t + b$). However, both of these invariances can be attained by simply normalizing the time series prior to analysis (computing Z-scores) as required by the application. This is required in most domains \cite{Rakthanmanon2012,Batista2013}, and in fact we do perform such a normalization when computing normalized distances in Section \ref{sec:analysis_collections}.

Other invariances must be established by modifying our $d_p$ distances, which should be possible by combining with existing approaches. For example, \cite{Batista2013} elegantly multiply the Euclidean metric by an ordered ratio of the complexities of two time series to establish a complexity-invariant distance. That approach should work analogously for our $d_p$ metrics. Phase invariance is a complex problem, for which typically the best approach is testing all possible alignments \cite{Keogh2008_DTW,uni2006}. A similar problem exists for (global) warping, namely (uniform) scaling in the $t$-direction: one must test all possibilities within a given range \cite{Keogh2003}. Such approaches could be incorporated for our $d_p$ distances.


\section{Methodology}
\label{sec:analysis_collections}

In this section, we describe a general procedure to use the function $\mathcal{F}$ and the induced distances $d_p$ to compute and analyze distances in a collection of time series. We will subsequently apply this procedure to analyze the Australian bushfires in the next section. Suppose $(X_t^{(1)}),...,(X_t^{(n)})$ are time series over some time period $[0,H]$. We suppress the time $t$ in our notation.

First, we choose an appropriate change point algorithm for the data, including threshold parameters, and a desired statistical quantity, such as mean or variance. This specifies a precise instance of $\mathcal{F}$. Let the piecewise constant functions associated to the time series $X^{(i)}$ be $f_i$, so that $\mathcal{F}(X^{(i)})=[f_i],i=1,...,n$. Also, we consider the normalized piecewise constant functions  $\hat{f_i}=\frac{f_i}{\|f_i\|_p}$. We form three different matrices:
\begin{align}
\Omega_{ij} = \frac{\langle\mathcal{F}(X^{(i)}), \mathcal{F}(X^{(j)})\rangle}{\|\mathcal{F}(X^{(i)})\|_2 \|\mathcal{F}(X^{(j)})\|_2}= \frac{\langle f_i, f_j\rangle}{\|f_i\|_2 \|f_j\|_2};\\
D^{us}_{ij}=d_p(X^{(i)},X^{(j)})= \|f_i-f_j\|_p; \\
D^{\text{norm}}_{ij}= \| \hat{f_i} - \hat{f_j} \|_p.
\end{align}
\begin{remark}
\label{remark:magnitude}
We recall $\|f_i\|_p=\magn(X^{(i)})$. Thus, the normalized $\hat{f_i}$ are produced by dividing by the aforementioned time series' magnitude functions.
\end{remark}
\begin{remark}
\label{remark:relationship}
We note a relation between $D^{\text{norm}}$ and $\Omega$ for $p=2$. If $u,v$ are two non-zero elements of an inner product space, with normalized elements $\hat{u},\hat{v}$ then
\begin{align}
\|\hat{u}-\hat{v}\|^2_2 = 2 - 2\langle\hat{u},\hat{v}\rangle = 2 - 2\frac{\langle u,v\rangle}{\|u\|_2 \|v\|_2},\\
\text{so } D^{\text{norm}}_{ij}=(2 - 2 \Omega_{ij})^\frac12.
\end{align}

\end{remark}
\begin{definition}
Let a $D$ and $A$ be $n \times n$ matrices. Then $D$ a termed a \emph{distance matrix} if $D$ is symmetric, $D_{ii}=0$ for all $i,$ and $D_{ij}\geq 0, \forall i,j$. Similarly, $A$ is termed an \emph{affinity matrix} if $A$ is symmetric, $A_{ii}=1,$ and $0 \leq A_{ij} \leq 1, \forall i,j.$
\end{definition}
Given a distance matrix $D$ one may produce an affinity matrix $A$ by: 
\begin{align}
\label{eq:affinitydefn}
A_{ij} = 1 - \frac{D_{ij}}{\max_{k,l} D_{kl}}.
\end{align}

$D^{us}$ is a distance matrix, thus named as it consists of unscaled distances. $D^{\text{norm}}$ is also a distance matrix, consisting of distances between normalized vectors, again with $p\geq 1$ suppressed from the notation. We refer to $\Omega$ as an \emph{alignment matrix} as it measures angles between the $f_i$ as vectors in an inner product space. We note that all entries of $\Omega$ lie in $[-1,1]$ with all diagonal elements equal to $1.$ To the distance matrices $D^{us}$ and $D^{\text{norm}}$ we associate affinity matrices $A^{us}$ and $A^{\text{norm}}$, respectively.

\subsection{Cross-contextual analysis}
\label{cross-contextual}
In many applications (including that of this paper), a collection of $n$ time series has additional structure. Each individual series $X^{(i)}$ contains numerical data one can compare with each other, but there may also exist contextual data among the indices $i$. For instance, in our specific analysis in the next section, $X^{(i)}$ is a time series of measurements taken at a particular measuring station; the relationships between these stations, such as their physical distance, influences the data. It is of interest to analyze not just the collection of time series, but their relationship with physical distance of these measuring stations, or other cross-contextual data. We introduce a general framework for cross-contextual distance analysis. Let $G$ be a distance matrix that codifies relationships between the sources of the data, $i=1,...,n.$ As above, we can associate an affinity matrix $A^G$ to these distances. Then, define \emph{consistency matrices}
\begin{align}
\text{Con}^{us}=A^{us} - A^G;\\
\text{Con}^{\text{norm}}=A^{\text{norm}} - A^G;\\
\text{Con}^\Omega = \Omega - A^G.
\end{align}
These matrices measure the consistency between the affinity measured by our $d_p$ measures on time series, and a cross-contextual distance matrix $G$.

\subsection{Clustering and outlier detection}
\label{clustering methodology section}
With the definitions of Section \ref{cross-contextual} in mind, we may apply the methods of spectral and hierarchical clustering in two ways. First, as is standard, we may apply it directly to the matrices $\Omega$, $D^{us}, D^{\text{norm}}$ to determine anomalies and clusters of similarity with respect to the $d_p$ distance between time series. Because the entries of $A$ are obtained as a linear combination of the entries of $D$, clustering with respect to a distance matrix or its associated affinity matrix gives identical results. The same applies for $\Omega$ and $D^{\text{norm},p=2}$, by remark \ref{remark:relationship}: these will cluster in identical ways.

In addition, we may also apply clustering techniques to the consistency matrices defined above. This will identify those anomalous data sources, indexed by $i$, which do not behave similarly between the two distance measures. In Section \ref{sec:bushfires}, we will uncover striking similarity when considering both distances between air quality time series and geographical distances and will identify just a handful of anomalous locations where this relationship does not hold; these are anomalies in the consistency matrix.

Finally, we also compute the matrix norms of the consistency matrices as a measure of the overall consistency between the distance measures. We use the normalized $L^1$ norm
\begin{align}
\|C\|=\frac{1}{n^2}\sum^n_{i,j=1} |c_{ij}|.
\end{align}

\section{Australian bushfire data}
\label{sec:bushfires}

In this section, we analyze AQI time series from $n=52$ locations across NSW \cite{bushfireAQIdata}. Each time series consists of hourly data from October 20, 2019 - January 20, 2020, to yield time series of length $H=2210$. In the rare case of missing data points, we take the value from the most recent prior hour. We record latitudes and longitudes of the measurement stations, from which we compute the matrix of geographical distances $G$ using the Haversine formula \cite{haversine}.

Mathematically, this data is a suitable opportunity to apply our $d_p$ metrics derived from $L^p$ distances between piecewise constant functions generated from change point algorithms. The change point procedure helps stabilize the data more appropriately, measuring average air quality over intervals of days rather than the huge peaks seen in the raw AQI time series. Since we expect air quality to remain somewhat locally stationary, these change point algorithms are more appropriate than more flexible but less targeted non-parametric models.

Within this multivariate context, we perform two separate analyses. First, our new $d_p$ distance matrices, computed according to the algorithmic framework of Section \ref{sec:proofs}, enables the identification of similarity clusters and anomalies. Second, calculating the geographical distances between measuring sites gives a contextual distance matrix $G$. As described in Section \ref{sec:analysis_collections}, we may form associated consistency matrices. With each of these matrices, we apply hierarchical and spectral clustering, thereby identifying anomalies in the spread of poor air relative to space and time.

Henceforth, we set $p=1$.

\subsection{Raw AQI data}
In Figure \ref{fig:AQITime}, we display the raw AQI time series data for select measuring stations across NSW, and perform our algorithm to fit piecewise constant functions for each via $\mathcal{F}_\mu$, as specified in Section \ref{sec:proofs}. Our results are consistent with geographical proximity and provide several preliminary insights regarding anomalous air quality relative to space and time during our period of analysis. A large number of these measurement sites are located in the city of Sydney, and hence close relative to the rest of the state of NSW. For instance, Randwick and Rozelle are two sites within Sydney that are particularly geographically close geographically, approximately 10 km apart. Figures \ref{fig:Randwick} and \ref{fig:Rozelle} illustrate the striking similarity between these air quality time series. The other four sites will be relevant in subsequent analysis. In Figure \ref{fig:Narrabri}, we note that Narrabri air quality was consistently poor from early November 2019. Wagga Wagga, Katoomba, and Albury, displayed in Figures \ref{fig:WaggaWagga}, \ref{fig:Katoomba}, and \ref{fig:Albury}, all suffered sharp peaks in poor air quality much later than Narrabri. Albury in particular exhibits peaks in poor air quality as late as mid January 2020. Indeed, Albury, on the southern border of NSW, was affected by the fires much later, as they moved south. Curiously, we see that the two Sydney-based stations, Randwick and Rozelle, exhibited rather moderate levels of pollution from mid-December onward, while Katoomba, located nearby, experienced its most severe air quality.

\begin{figure}
    \centering
        \begin{subfigure}[b]{0.49\textwidth}
        \includegraphics[width=\textwidth]{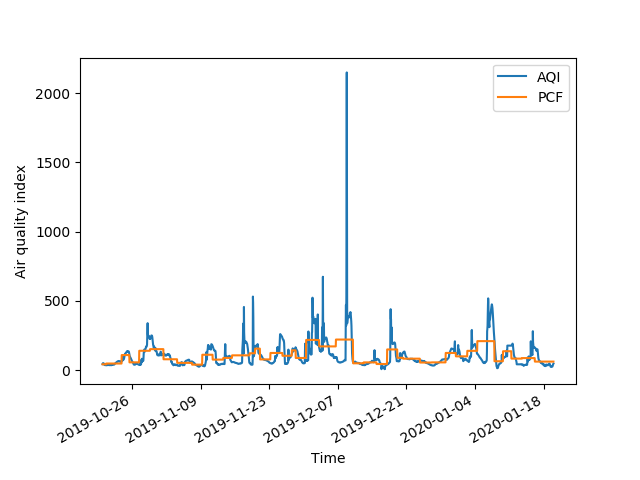}
        \caption{}
        \label{fig:Randwick}
    \end{subfigure}
    \begin{subfigure}[b]{0.49\textwidth}
        \includegraphics[width=\textwidth]{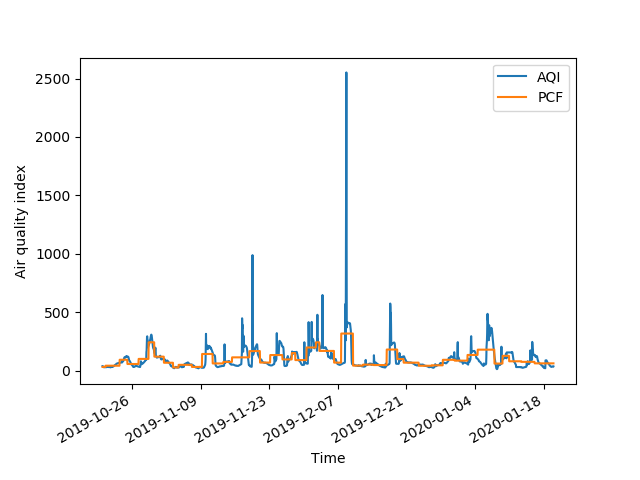}
        \caption{}
        \label{fig:Rozelle}
    \end{subfigure}
    \begin{subfigure}[b]{0.49\textwidth}
        \includegraphics[width=\textwidth]{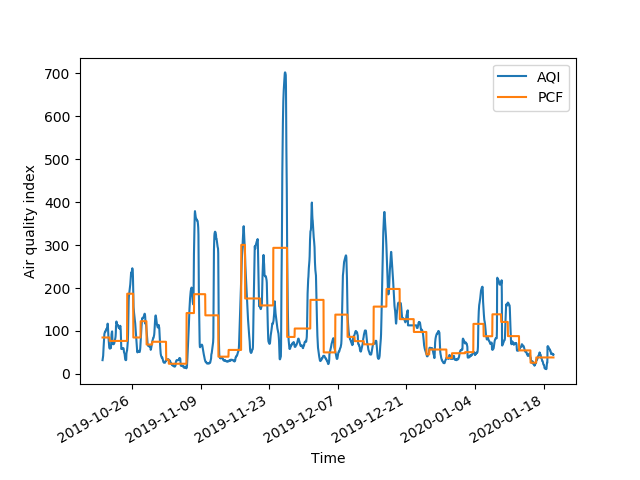}
        \caption{}
        \label{fig:Narrabri}
    \end{subfigure}
    \begin{subfigure}[b]{0.49\textwidth}
        \includegraphics[width=\textwidth]{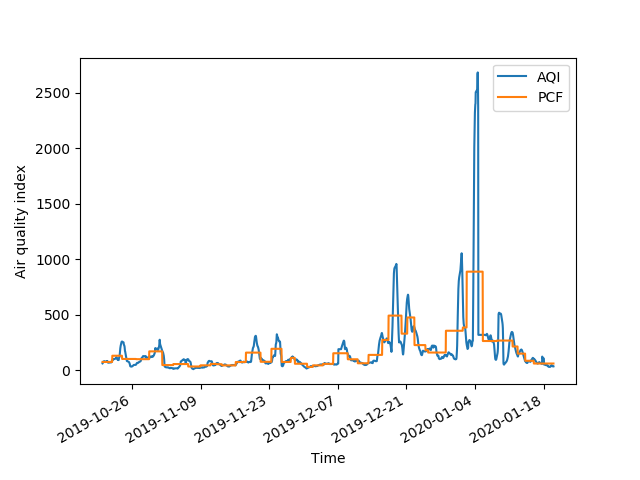}
        \caption{}
        \label{fig:WaggaWagga}
    \end{subfigure}
    \begin{subfigure}[b]{0.49\textwidth}
        \includegraphics[width=\textwidth]{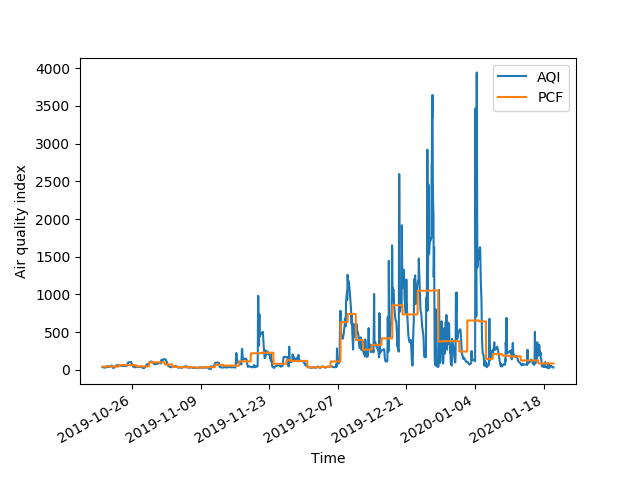}
        \caption{}
        \label{fig:Katoomba}
    \end{subfigure}
     \begin{subfigure}[b]{0.49\textwidth}
        \includegraphics[width=\textwidth]{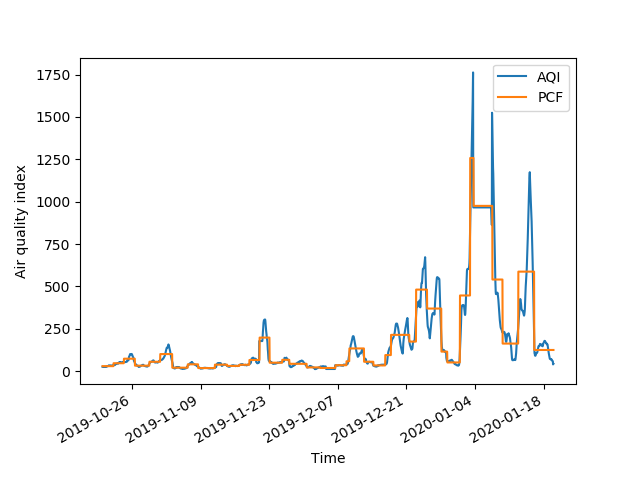}
        \caption{}
        \label{fig:Albury}
    \end{subfigure}
\caption{AQI time series and associated piecewise constant function for (a) Randwick (b) Rozelle (c) Narrabri (d) Wagga Wagga (e) Katoomba (f) Albury. We see that the two Sydney-based sites in (a) and (b) exhibited rather similar air quality profiles, with greater severity before mid-December. The more distant southern sites Wagga Wagga (d) and Albury (f) experience greater severity toward the end of the period. Curiously, Katoomba (e), very close to Sydney, resembles the distant southern sites, while Narrabri (c) resembles the Sydney sites despite being quite far to the north. We can already see some geographic irregularity in these time series.
}
\label{fig:AQITime}
\end{figure}

\subsection{Affinity matrix analysis}
\label{sec:affinitymatrixanalysis}
In this section, we use our new distance measures, via the function $\mathcal{F}_{\mu}$, to produce clusters of similarity and anomalies with respect to the air quality of different locations. As remarked in Section \ref{clustering methodology section}, there is no difference between clustering on a distance matrix or its affinity matrix, so we cluster based on affinity matrices $A^{us}$ and $A^{\text{norm}}$. These entries lie in $[0,1]$ so that diagrams display a consistent scale. These matrices each provide interesting insights, and are largely consistent with respect to their identification of similarity clusters and anomalies. 

\begin{figure}
    \centering
    \begin{subfigure}[b]{0.765\textwidth}
        \includegraphics[width=\textwidth]{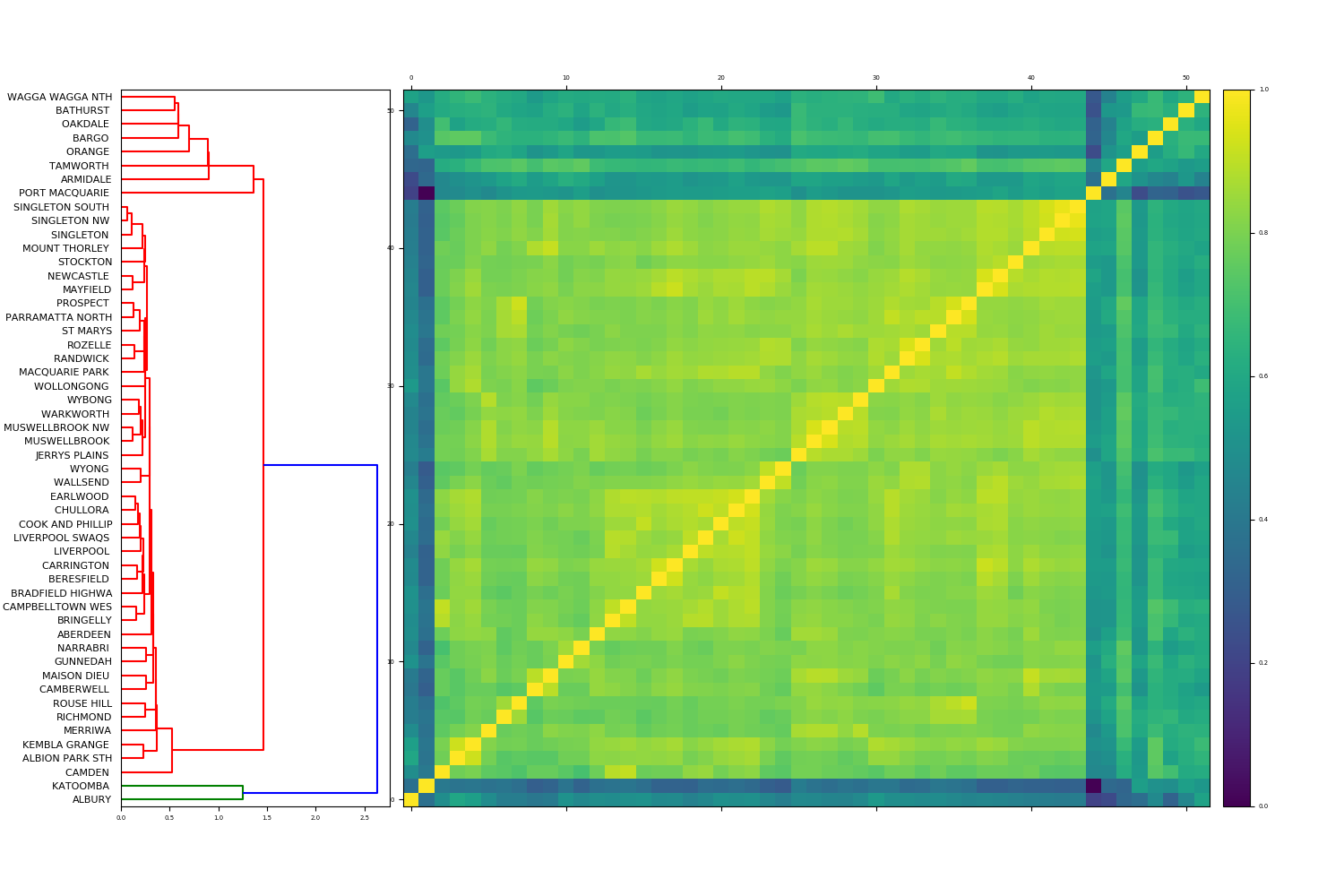}
        \caption{}
        \label{fig:UnscaledAffinity}
    \end{subfigure}
    \begin{subfigure}[b]{0.765\textwidth}
        \includegraphics[width=\textwidth]{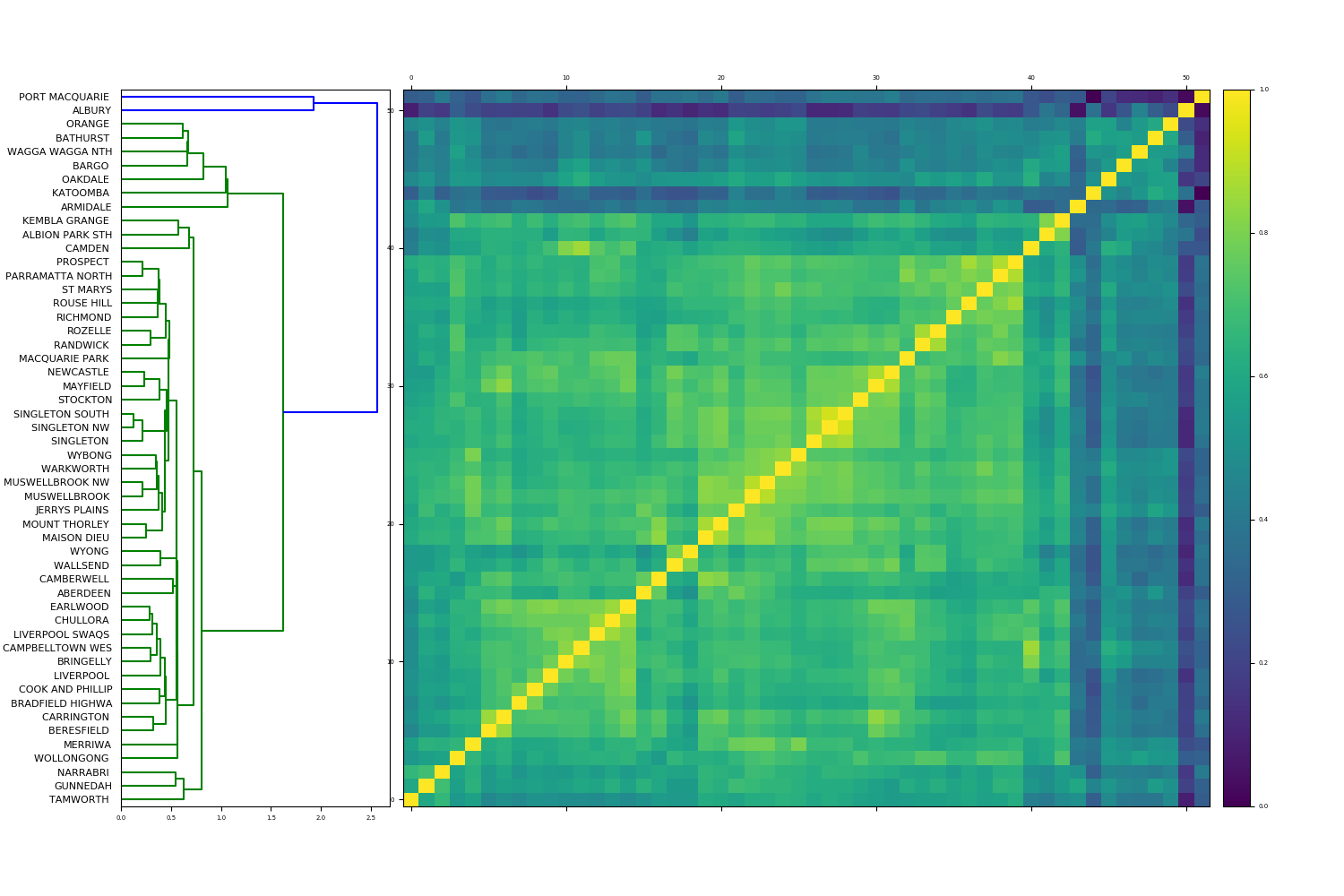}
        \caption{}
        \label{fig:NormAffinity}
    \end{subfigure}
    \begin{subfigure}[b]{0.765\textwidth}
        \includegraphics[width=\textwidth]{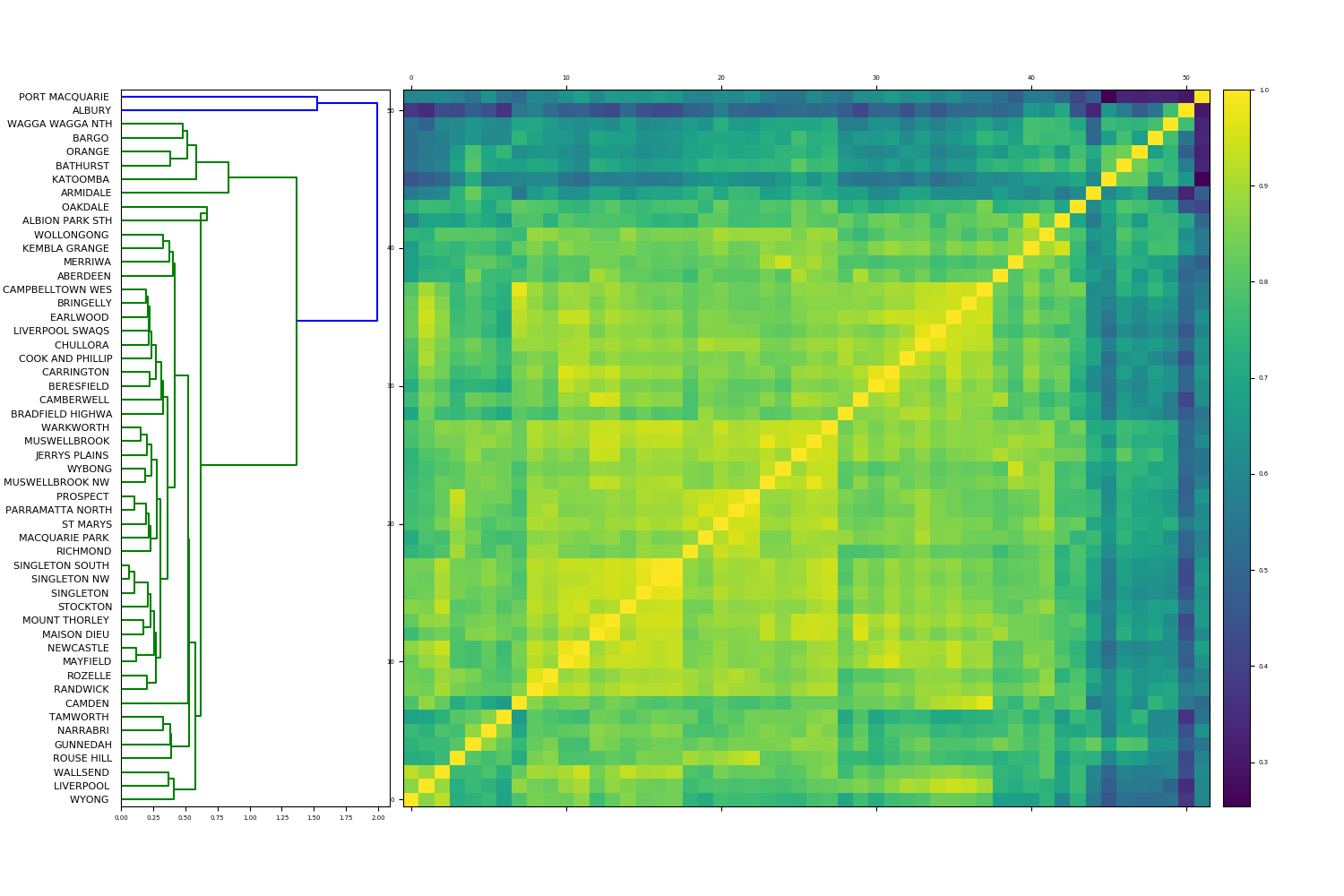}
        \caption{}
        \label{fig:Correlation}
    \end{subfigure}
\caption{Hierarchical clustering on three affinity matrices defined in Section \ref{sec:analysis_collections}, (a) $A^{us}$ the unscaled distance matrix (b) $A^{\text{norm}}$ the normalized distance matrix (c) $\Omega$ the alignment matrix. As remarked in Section \ref{sec:affinitymatrixanalysis}, Katoomba leaves the most anomalous cluster when we pass from the unscaled distances to the normalized distances. This is because it has the single greatest magnitude of all locations, so taking normalized distances nullifies this extreme feature.
}
\label{fig:DistanceMatrixDendrogram}
\end{figure}

\begin{table}
 \begin{tabular}{|c c| c c|} 
 \hline
 Measuring station & Time series magnitude & Measuring station & Time series magnitude \\
 \hline
 Aberdeen & 112.1 & Mount Thorley & 122.2\\ 
 \hline
 Albion Park & 89.3 & Muswellbrook & 126.0\\ 
 \hline
 Albury & 175.2 & Muswellbrook NW & 125.5 \\ 
 \hline
 Armidale & 195.5& Narrabri & 105.7\\ 
 \hline
 Bargo & 152.7 &   Newcastle & 112.8\\ 
 \hline
 Bathurst & 181.9 &  Oakdale & 196.6\\ 
 \hline
 Beresfield & 110.4 &  Orange & 204.2 \\ 
 \hline
 Bradfield Highway & 104.5&   Parramatta North & 116.1 \\ 
 \hline
 Bringelly & 125.6&   Port Macquarie & 202.0 \\ 
 \hline
 Camberwell & 139.0 & Prospect & 123.9 \\ 
 \hline
 Camden & 140.9&   Randwick & 104.6 \\ 
 \hline
 Campbelltown West & 123.5 &  Richmond & 137.7  \\ 
 \hline
 Carrington & 107.0&  Rouse Hill & 131.2 \\ 
 \hline
 Chullora & 108.9 &   Rozelle & 102.9 \\ 
 \hline
 Cook and Phillip & 100.7 &   Singleton & 116.9 \\ 
 \hline
 Earlwood & 103.0  &    Singleton NW & 116.5\\ 
 \hline
 Gunnedah & 104.6&  Singleton South & 113.0\\ 
 \hline
 Jerrys Plains & 131.3 & St Marys & 121.4\\ 
 \hline
 Katoomba & 258.4&  Stockton & 133.2\\ 
 \hline
 Kembla Grange & 97.7 &  Tamworth & 159.8 \\ 
 \hline
 Liverpool & 116.3 &  Wagga Wagga North & 174.7\\ 
 \hline
 Liverpool Swaqs & 121.8 & Wallsend & 100.0 \\ 
 \hline
 Macquarie Park & 105.1 &   Warkworth & 133.5\\ 
 \hline
 Maison Dieu & 140.7&   Wollongong & 98.7\\ 
 \hline
 Mayfield & 110.2 &     Wybong & 122.7\\
 \hline
Merriwa & 133.1&  Wyong & 108.3 \\
 \hline
\end{tabular}
\caption{Measuring station names and time series magnitudes $\magn$, as defined in Section \ref{sec:proofs}, reflecting an overall measure of the noise-adjusted average air quality index. We can see that Katoomba exhibits the single greatest magnitude, which contributes to its outlier status in Figure \ref{fig:UnscaledAffinity}. Curiously, the site with the third-greatest magnitude, Port Macquarie, shows up as an anomalous in Figure \ref{fig:NormAffinity} and \ref{fig:Correlation}, which used normalized distance measures.
}
\label{tab:magnitudes}
\end{table}

Figure \ref{fig:UnscaledAffinity} displays hierarchical clustering of $A^{us}$, our affinity matrix between unscaled distances. The dendrogram indicates that there is a majority cluster of time series, with two outlier elements, Katoomba and Albury, both of which had marked peaks in poor air quality, confirmed by Figures \ref{fig:Katoomba} and \ref{fig:Albury}. Both areas experienced some of their worst air quality well after other NSW measuring stations, leading to their piecewise constant function being classified as anomalous. The majority cluster contains two subgroups. The first of these is a large subcluster of highly similar time series, containing all Sydney-based measuring stations. The second subcluster consists of less pronounced anomalies, including Armidale, Bargo, Bathurst, Oakdale, Orange, Tamworth, Wagga Wagga (\ref{fig:WaggaWagga}), and Port Macquarie. All of these locations are ``regional locations'' located a significant distance from the city of Sydney. Spectral clustering results suggest that there are three clusters: one containing Albury, another containing Katoomba, and the third consisting of the remainder of the time series.

Figure \ref{fig:NormAffinity} depicts $A^{\text{norm}}$, our affinity matrix recording distances between normalized vectors. We have similar results for this matrix. Hierarchical clustering again reveals a large majority cluster, with two outliers, Port Macquarie and Albury. The majority cluster has two subgroups again. One is smaller, consisting of Armidale, Bargo, Bathurst, Katoomba (\ref{fig:Katoomba}), Oakdale, Orange, Tamworth, and Wagga Wagga (\ref{fig:WaggaWagga}), while the other subgroup contains all remaining locations. In particular, we note that after normalization, Katoomba has left the most anomalous cluster. Indeed, Katoomba is characterized by the single greatest magnitude of all locations, with $\magn(X_t) \approx 258$. We record all time series magnitudes in Table \ref{tab:magnitudes}. After normalization by the magnitude, recalling Remark \ref{remark:magnitude}, this extreme feature is nullified. The remainder of the measuring station time series, which again contain all of Sydney, are identified as highly similar. Spectral clustering on the $A^{\text{norm}}$ matrix indicates that there are two highly anomalous time series, Port Macquarie and Albury.

Third, we analyze the alignment matrix $\Omega$. Hierarchical clustering is displayed in Figure \ref{fig:Correlation}. These results are almost identical to the hierarchical clustering on $A^{\text{norm}}$, identifying Port Macquarie and Albury as primary anomalies, and Armidale, Bargo, Bathurst, Katoomba, Oakdale, Orange, Tamworth, and Wagga Wagga as secondary anomalies. The remaining stations, once again containing the entire city of Sydney, are closely clustered. Spectral clustering proposes three clusters of AQI: one cluster containing Port Macquarie, one containing Albury, and one remaining cluster consisting of the remaining locations.

\subsection{Cross-contextual analysis}
In this section, we analyze the consistency matrices defined in Section \ref{sec:analysis_collections}. These consistency matrices provide a framework for cross-contextual analysis. That is, where one matrix computes the similarity in one aspect, and another matrix computes similarity in a separate aspect, subtracting one affinity matrix from the other allows us to identify the consistency between different similarity measures across a collection of time series. This helps to place measurements between different time series within a greater context. In this instance, we analyze whether air quality is consistent with respect to geographical distance between measuring locations. Anomalies in consistency highlight either a) two areas close by way of geographical distance but dissimilar in air quality or b) two locations far apart in terms of geographical distance, but similar in their air quality. 

Hierarchical clustering on the geographical affinity matrix $A^G$ (see Figure \ref{fig:ConsistencyMatrixDendrogram}) determines clusters of measuring stations in terms of their geographical proximity within the collection. A majority cluster is identified with two outliers, Wagga and Wagga, both of which are significantly further south than the other measuring stations. The remaining measuring stations fall in one cluster that appears to have two distinct collections of similarity. The first sub-group contains the city of Sydney sites, while the second sub-group includes locations outside Sydney, both to the west and north.

First, we analyze the consistency matrix associated to the unscaled distance measure, $\text{Con}^{us}$. Hierarchical clustering (Figure \ref{fig:ConsistencyUnscaled}) indicates a majority cluster with one outlier location, Katoomba. Within the majority cluster there is a small subcluster of anomalous locations including Gunnedah, Narrabri, Albury, Wagga Wagga and others. That is, Katoomba is highly anomalous relative to the relationship between air quality and distance from other locations. Indeed, Katoomba is relatively close to the city of Sydney but was befallen with terrible bushfires close by. Spectral clustering indicates that there are three clusters: the first consisting solely of Katoomba, the second consisting solely of Oakdale and the final cluster containing the remaining measuring locations. The norm of the entire consistency matrix is approximately $0.11,$ indicating broad similarity between geographical and AQI affinity.

Next, we consider the consistency matrix associated to the distances between normalized vectors, $\text{Con}^{\text{norm}}$. The hierarchical clustering results (Figure \ref{fig:ConsistencyNorm}) are highly similar to that of $\text{Con}^{us}$, with a majority cluster detected but Katoomba as a single outlier. Interestingly, there is more variance in the structure of the majority cluster. This demonstrates the utility of our plurality of measurements: others may uncover structure that one alone cannot. Katoomba's outlier status in this second experiment is by no means obvious and provides another insight: even after normalizing by its extreme magnitude, it is still anomalous relative to the consistency between geographical distance and AQI. Spectral clustering indicates that there are two anomalous locations, Katoomba and Bargo, each of which is contained exclusively within its own cluster. Finally, the matrix norm is $\approx 0.19$, larger than that of $\text{Con}^{us}$. This makes sense given less consistency overall in this matrix.

Finally, we consider the consistency matrix associated to the alignment matrix, $\text{Con}^\Omega$. Hierarchical clustering (Figure \ref{fig:ConsistencyCorrelation}) indicates that there are three clusters of locations. The first is an outlier element, Narrabri, the second cluster consists of Bargo and Katoomba, and the final majority cluster consists of the remaining measurement stations. These can be subcategorized into those located closer to Sydney and those located farther from Sydney. Spectral clustering indicates that the two most anomalous locations are Katoomba and Bargo. Given the reasonable difference in inference generated from each consistency matrix, one can appreciate the importance of using the most appropriate distance measure. The matrix norm is $\approx 0.10$, once again indicating broad similarity. Again, the fact that $\text{Con}^\Omega$ identifies Narrabri as an outlier while $\text{Con}^{us}$ does not demonstrates the utility of our range of measurements, despite the fact that these two distance methods are quite related by Remark \ref{remark:relationship}.

\begin{figure}
    \centering
    \begin{subfigure}[b]{0.49\textwidth}
        \includegraphics[width=\textwidth]{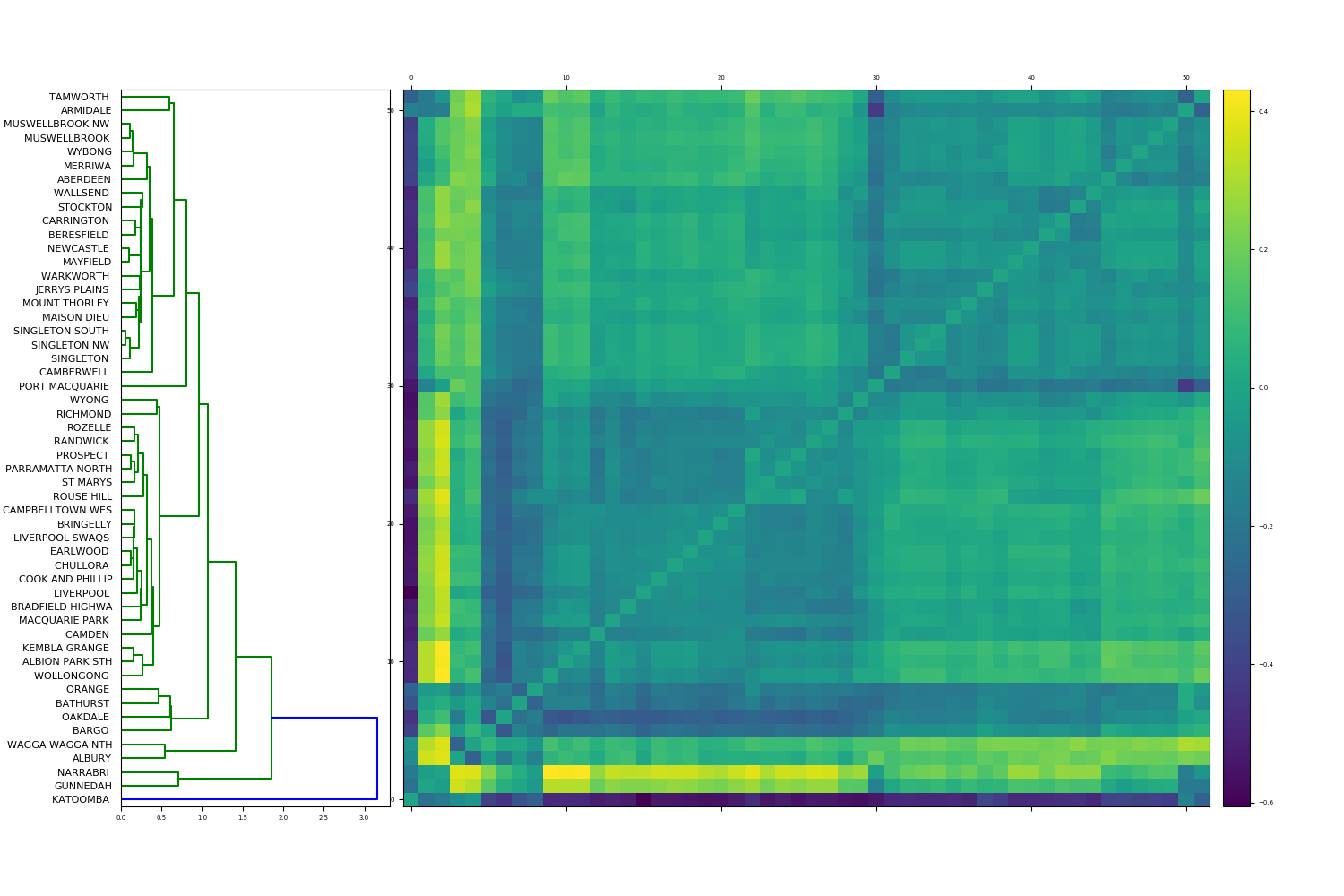}
        \caption{}
        \label{fig:ConsistencyUnscaled}
    \end{subfigure}
    \begin{subfigure}[b]{0.49\textwidth}
        \includegraphics[width=\textwidth]{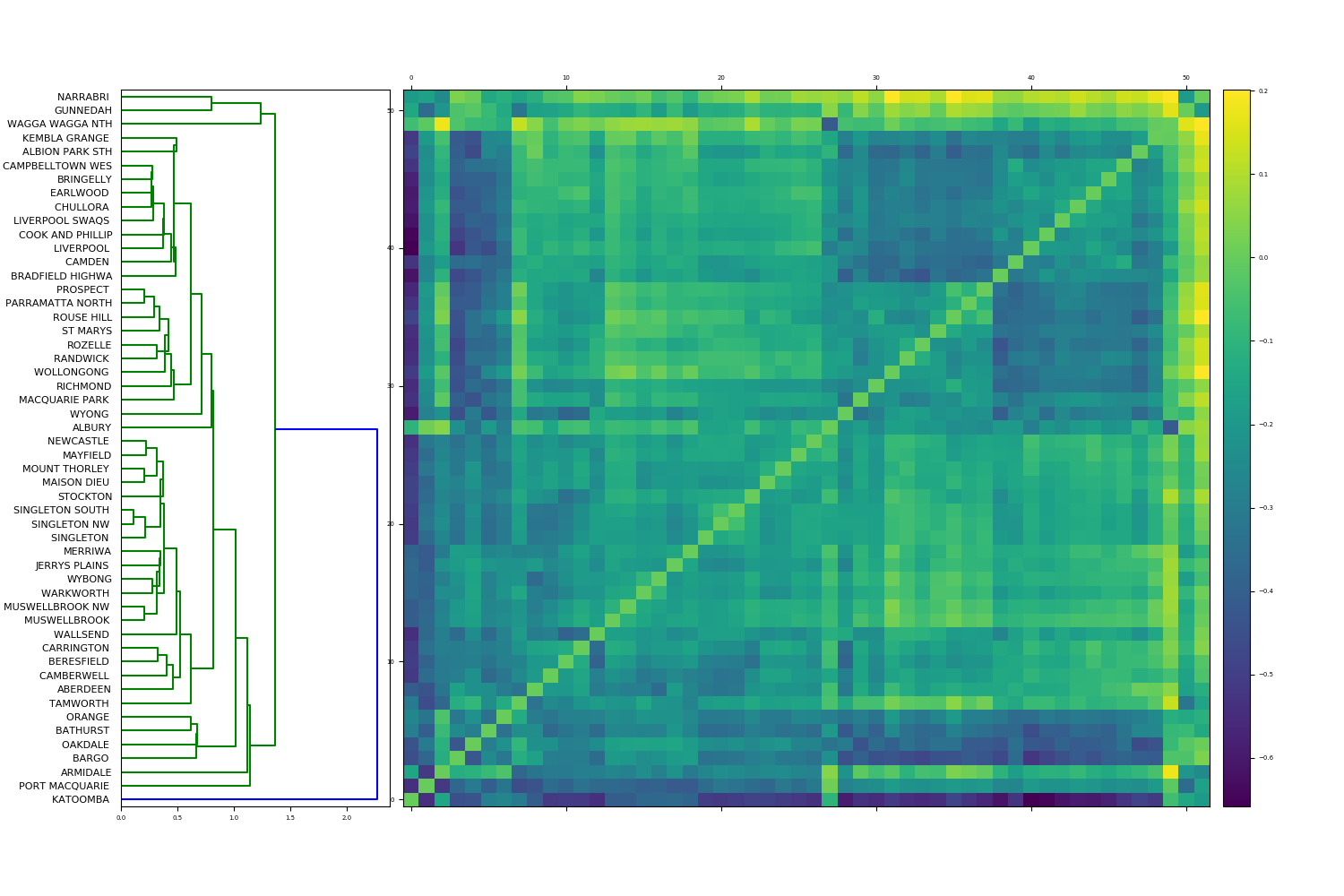}
        \caption{}
         \label{fig:ConsistencyNorm}
    \end{subfigure}
        \begin{subfigure}[b]{0.49\textwidth}
        \includegraphics[width=\textwidth]{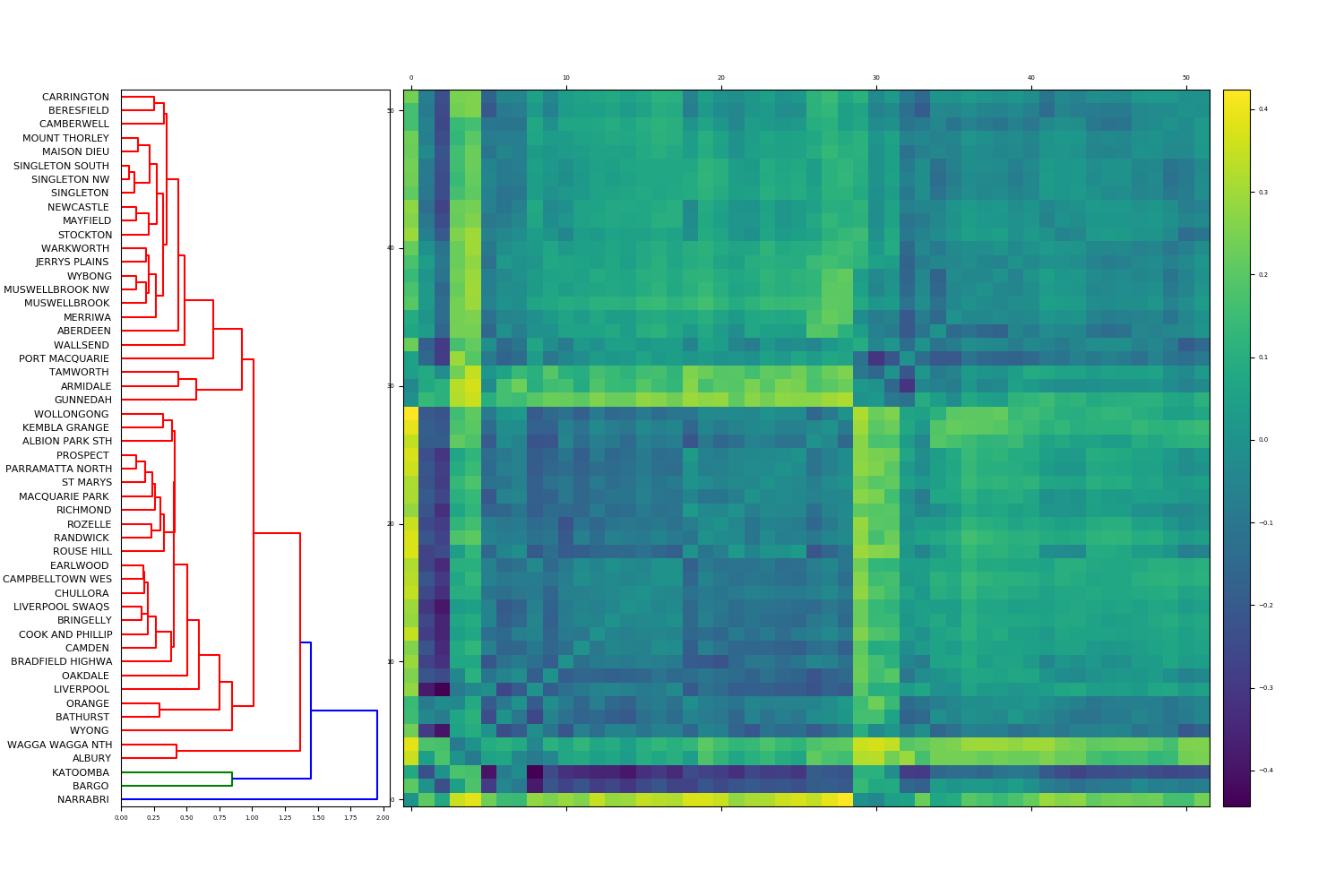}
        \caption{}
        \label{fig:ConsistencyCorrelation}
    \end{subfigure}
    \begin{subfigure}[b]{0.49\textwidth}
        \includegraphics[width=\textwidth]{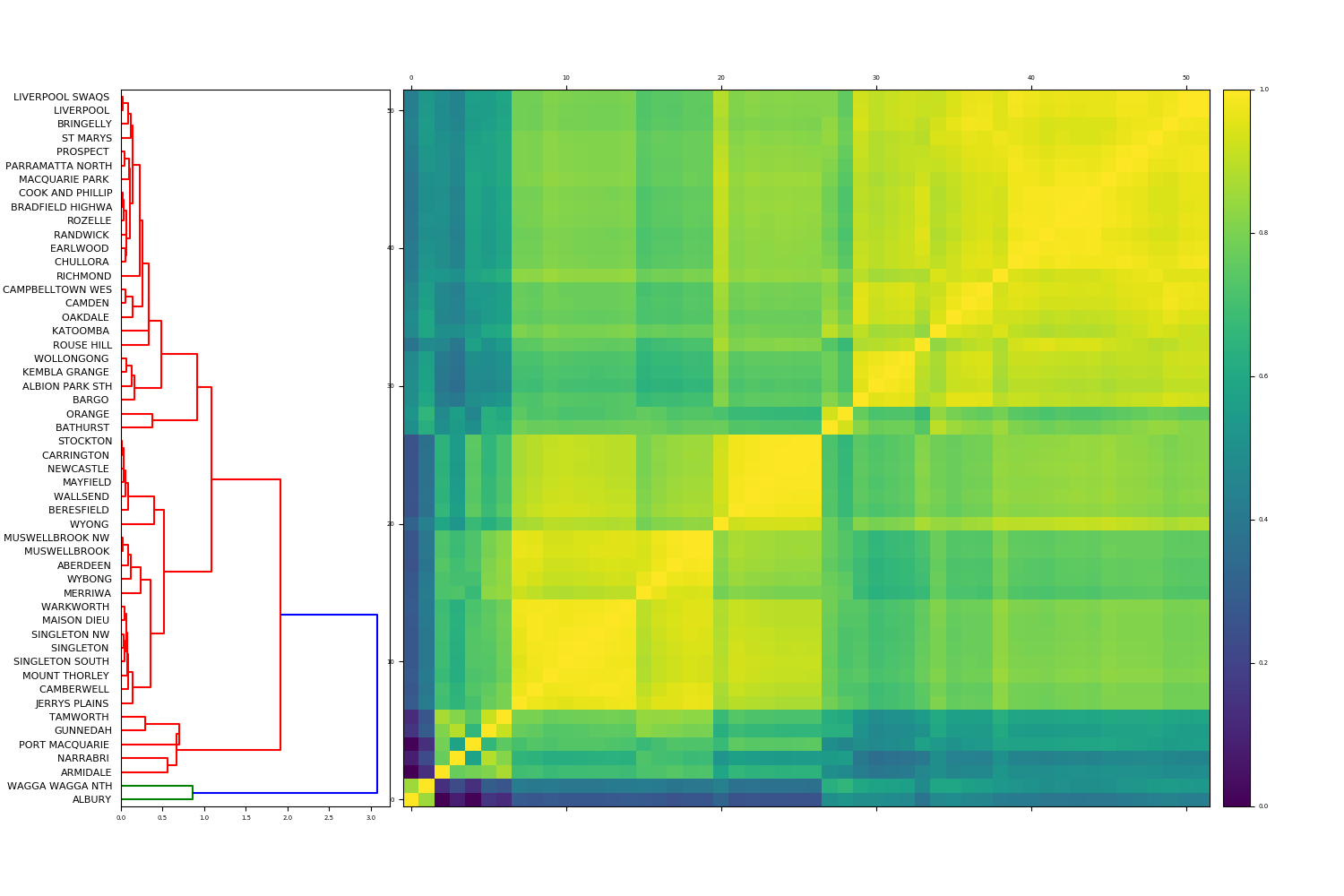}
        \caption{}
        \label{fig:Geodesicdendrogram}
    \end{subfigure}
\caption{Hierarchical clustering on three consistency matrices defined in Section \ref{sec:analysis_collections}, (a) $\text{Con}^{us}$ the unscaled consistency matrix (b) $\text{Con}^{\text{norm}}$ the normalized consistency matrix (c) $\text{Con}^{us}$  the alignment consistency matrix and (d) the geographical distance matrix $G$. Numerous insights are visible here. First, Katoomba is an outlier element for both the unscaled and normalized distances. Its outlier using the normalized distances is by no means obvious, and reflects Figure \ref{fig:AQITime} where Katoomba was highly inconsistent with geographically nearby measuring sites, not just in magnitude of the time series but in normalized profile as well. Specifically, Katoomba exhibited its greatest severity much later in the period of analysis than nearby Sydney sites. Second, Narrabri is revealed as an outlier element in (c), also reflecting what we saw in Figure \ref{fig:AQITime}. Notably, this is the first time Narrabri has been identified as an outlier anywhere in our analysis, and the fact that it is only revealed in (c) but not the closely related (b) demonstrates the utility in the plurality of measurements.
} 
\label{fig:ConsistencyMatrixDendrogram}
\end{figure}

\section{Conclusion}
\label{sec:conclusion}

We have proposed a new theoretical and practical framework for understanding equivalence and computing distances between time series. Mathematically, we aim to identify and quotient by the relation that two time series differ only up to noise, and computationally, we want to remove noise from calculations of distances $d_p$ and magnitudes $\magn$. Our procedure builds on earlier work to measure distances between time series via sets of structural breaks, and our new method produces more reliable and continuous measures. We have applied the function $\mathcal{F}_\mu$ and these metrics to analyze data with extraordinary peaks, determining piecewise constant functions $f$ designed to understand the true movement of the mean (or variance) throughout the time series. This procedure is highly flexible: we could swap variance for mean, vary our change point algorithm parameters, or use completely different change point algorithms altogether.

Several future implementation opportunities exist in this equivalence framework. One could develop and explore notions of \emph{weak equivalence} if two time series are equivalent under $\mathcal{F}_{\mu}$ or $\mathcal{F}_{var}$, and \emph{strong equivalence} if they are equivalent under both, or more strongly, if they have identical distributions within their change point segments. Moving beyond piecewise constant functions, two autoregressive progresses with the same parameters $\phi_t = A \phi_{t-1}+w_t$ and $\psi_t = A \psi_{t-1}+w'_t$ could be declared equivalent, if $w_t$ and $w'_t$ are each ``white noise'' variables, namely $N(0,\epsilon^2)$ Gaussian distributions for sufficiently small $\epsilon$. Algorithms could be developed to empirically detect such similarity in the distributions from an observed sample set, define equivalence relations, and hence produce reliable distance measures on the induced quotient space. Developing further theoretical and algorithmic ways of determining quotient spaces equipped with metrics may have many applications in determining appropriately adjusted measures between time series. Within a statistical or nonlinear dynamics context, algorithmic frameworks designed to learn appropriate notions of equivalence ``up to noise'' would also be of great interest to the community. Finally, our distance measures could be incorporated with other advances in distance analysis between time series, such as the search for various different invariances (amplitude, phase, complexity, and others) as discussed at the end of Section \ref{sec:simulation}.

In our analysis of the NSW bushfires, we have determined anomalies not just in the air quality data itself, but also those sites that stand out in the otherwise strong consistency between air quality and geographical distances. By using several different discrepancy matrices, we are able to uncover structure that one alone cannot reveal. Katoomba is a location of particular interest in our analysis. It was initially revealed as highly anomalous in our unscaled distance matrix $D^{us}$ but less so in the normalized matrix $D^{\text{norm}}$. The initial logical explanation for this is that the feature most anomalous with Katoomba was its considerable scale, reflected in the location's overall magnitude, the highest among all measuring stations, which is then nullified when normalizing distances are computed in $D^{\text{norm}}$. However, curiously, Katoomba again appears as an outlier in $\text{Con}^{\text{norm}}$, which uses normalized distances to measure consistency with geographical distance. Examining Katoomba's time series again (Figure \ref{fig:Katoomba}) we note that Katoomba experienced its greatest severity in air quality later on in our period of analysis, from mid-December onward, whereas the Sydney stations (which Katoomba is relatively close to) generally experienced severity prior to this. 

There is also utility in examining the results of a range of clustering methods. For example, only spectral clustering identifies Oakdale as an anomalous element in its own cluster when considering unscaled distances' consistency. We have also identified Narrabri as an interesting case. Unlike Katoomba, this location was never identified as a significant anomaly with respect to geography or AQI, considering distances or alignment. And yet, it is anomalous in the consistency between geographical and AQI distances. It is one of very few locations where these two measures are not closely aligned relative to other sites. Indeed, in Figure \ref{fig:Narrabri} we can see it closely resembles the profile of the Sydney sites despite being quite distant, to the north of the state. In addition, the holistic analysis of various distance and consistency matrices yields several high-level insights. First, there is broad consistency between AQI and geographical affinity matrices, which measure the relationship between sites with respect to these two aspects, as evidenced with the low consistency matrix norms in all implementations. Second, there is clearly sensitivity in the anomalies detected based on the distance measures that are used.

The mathematical analysis of relationships and anomalies between air quality measuring stations could prompt other researchers to conduct closer meteorological or geological examinations of different locations. Such research could shed light on the anomalous nature of Katoomba, Narrabri and other sites discussed in this manuscript. The fact that Sydney's air quality was much more moderate after mid-December while nearby Katoomba's became much more severe are too persistent with time to be explained by coincidences such as wind direction, and further exploration of this reasons could provide further insight on how bushfires and air pollution spread. This and other mathematical analysis could encourage researchers to notice hereto-unknown features of the land and its susceptibility or resistance to extreme bushfires, allowing better allocation of resources and understanding of vulnerability to bushfires. As such extreme events increase, we urge more interdisciplinary work between nonlinear dynamics and environmental researchers into such questions.

\section*{Data availability statement}
The data analyzed in this study are publicly available at \cite{bushfireAQIdata}. A cached copy is conveniently available at \url{https://github.com/MaxMenzies/BushfireData}.

\appendix
\section{Change point algorithm}

In this appendix, we describe the specific implementation of the change point detection algorithm used in Sections \ref{sec:analysis_collections} and \ref{sec:bushfires}. The general change point detection framework is as follows: a sequence of observations $x_1,x_2,...,x_n$ are drawn from random variables $X_1, X_2,...,X_n$ and undergo an unknown number of changes in distribution at  $\tau_1,...,\tau_m$. We assume observations are independent and identically distributed between change points, that is, between each change points a random sampling of an underlying distribution is taken. Following \cite{RossCPM}, we notate this as follows
\begin{align}
    X_{i} \sim 
    \begin{cases}
      F_{0} \text{ if } i \leq \tau_1 \\
      F_{1} \text{ if } \tau_1 < i \leq  \tau_2  \\
      F_{2} \text{ if } \tau_2 < i  \leq \tau_3,  \\
      \hdots
    \end{cases}
\end{align}
While the requirement of independence may appear restrictive, dependence can generally be accounted for by several means, such as modelling the underlying dynamics or drift process, then applying a change point algorithm to the model residuals or one-step-ahead prediction errors, as described by \cite{gustafsson2001}.

\subsection{Batch detection (phase I)}
The first phase of change point detection is retrospective. As above, we are given a finite sequence of observations $x_1,\ldots,x_n$ from random variables $X_1,\ldots,X_n$. For simplicity, we assume there is at most one change point. If a change point exists at time $k$, then observations have some distribution of $F_0$ prior to the change point, and a distribution of $F_1$ proceeding the change point, where $F_0 \neq F_1$. Thus, one must test between the following two hypotheses for each $k$: 

\begin{align}
    H_0: X_{i} \sim F_0, i = 1,...,n
\end{align}

\begin{align}
    H_1: X_{i} \sim 
    \begin{cases}
      F_{0} & i = 1,2,...,k \\
      F_{1}, & i = k + 1, k+2, ..., n  \\
    \end{cases}
\end{align}
and select the most suitable point $k$, if one exists.

One proceeds with a two-sample hypothesis test, where the choice of test depends on assumptions about the underlying distributions. Nonparametric hypothesis tests can be chosen to avoid distributional assumptions. One appropriately chooses a two-sample test statistic $D_{k,n}$ and a threshold parameter $h_{k,n}$. If $D_{k,n}>h_{k,n}$ then the null hypothesis is rejected and one provisionally assumes that a change point has occurred at $x_k$. These test statistics $D_{k,n}$ are normalized to have mean $0$ and variance $1$ and evaluated at all values $k=1,...,n$; the largest test statistic is assumed to be coincident with the existence of our sole change point. The overall test statistic is then
\begin{align}
    D_{n} = \max_{k=2,...,n-1} D_{k,n} = \max_{k=2,...,n-1} \Bigg| \frac{\Tilde{D}_{k,n} - \mu_{\Tilde{D}_{k,n}}}{\sigma_{\Tilde{D}_{k,n}}}  \Bigg|
\end{align}
where $\Tilde{D}_{k,n}$ are non-normalized statistics.

The null hypothesis of no change point is rejected if $D_{n} > h_n$ for an appropriately chosen threshold $h_n$. In this case, we conclude that a (single) change point has occurred and its location is the value of $k$ that maximizes $D_{k,n}$. That is,
\begin{align}
    \hat{\tau} = \argmax_k D_{k,n}.
\end{align}
This threshold $h_n$ is chosen to bound the type 1 (false positive) error rate as is commonplace in statistical hypothesis testing. One specifies an acceptable level $\alpha$ for the proportion of false positives, the probability of falsely declaring that a change has occurred when in fact it has not. Then, $h_n$ is chosen as the top $\alpha$ quantile of the distribution of $D_n$ under the null hypothesis \cite{RossCPM}. 


\subsection{Sequential detection (phase II)}
In this phase, the sequence $(x_t)_{t \geq 1}$ does not have a fixed length. New observations are continually received over time, and multiple change points may exist. Assuming no change point exists so far, this phase treats $x_1,..., x_t$ as a fixed-length sequence and computes $D_t$ as described in phase I. A change is detected if $D_t > h_t$ for an appropriately chosen threshold. If no change is detected, the next observation $x_{t+1}$ is brought into the sequence $x_1,...,x_t,x_{t+1}$ for consideration and the process continues. If a change is detected, the process restarts from the data point immediately following the detected change point. Thus, the procedure consists of a repeated sequence of hypothesis tests.

In this sequential setting, $h_t$ is selected so that the probability of incurring a Type 1 (false positive) error is constant over time, so that under the null hypothesis of no change point, the following holds:
\begin{align}
    P(D_1 > h_1) &= \alpha,\\
    P(D_t > h_t | D_{t-1} \leq h_{t-1}, ... , D_{1} \leq h_{1}) &= \alpha, \ t > 1.
\end{align}
In this case, assuming that no change point occurs, the expected number of observations received before a false positive detection occurs is $\frac{1}{\alpha}$. This quantity is known as the average run length, or ARL$_0$ \cite{RossCPM}.

\section{Cluster theory}
In this section, we briefly describe the two methods of clustering used throughout this paper. First, \emph{hierarchical clustering} \cite{Ward1963,Szekely2005} is an iterative clustering technique that seeks to build a hierarchy of similarity between elements. In this paper we use agglomerative hierarchical clustering with the average linkage method \cite{Mllner2013}, where each element (in our case a location) begins in its own cluster and branches between them are successively built. The results of hierarchical clustering are commonly displayed in \emph{dendrograms}, as we display in this paper. Unlike spectral or K-means clustering, hierarchical clustering does not specify a discrete partition of elements, and does not require the choice of a number of clusters $k$. Like spectral clustering, hierarchical clustering can be implemented on any distance, not necessarily Euclidean space.

Spectral clustering applies a graph theoretic interpretation of a distance matrix $D$, projecting data into a lower-dimensional space, the eigenvector domain, where it may be more easily separated by standard algorithms such as $K$-means. Following \cite{vonLuxburg2007}, we transform a distance matrix $D$ into an affinity matrix $A$ as we do throughout the paper (\ref{eq:affinitydefn}). Then we define the graph Laplacian by
\begin{equation}
    L = E - A, 
\end{equation}
where $E$ is the diagonal degree matrix with diagonal entities  $E_{ii} = \sum_{j} A_{ij}$ and zero off-diagonal entries. $A$ and hence $L$ are real symmetric matrices, so can be diagonalized with all real eigenvalues  (the spectral theorem, \cite{Axler}). By construction, $L$ is \emph{positive semi-definite} with eigenvalues $0=\lambda_{1} \leq \lambda_{2} \leq ... \leq \lambda_{n}.$

Spectral clustering proceeds as follows: with $k$ chosen \textit{a priori}, we find corresponding eigenvectors $f_{1}, f_{2},...,f_{k}$ and construct the $n \times k$ matrix $F$ whose columns are $f_i,i=1,\dots,k$. Let $v_j \in \mathbb{R}^k$ be the rows of $F, j=1,\dots,n$. We apply standard K-means \cite{Lloyd1982} to cluster these rows into clusters $C_1,...,C_k$. Finally, we determine clusters $A_l=\{i: v_i \in C_l \}, l=1,...,k$ to assign the original $n$ elements into the corresponding clusters. Spectral clustering can be applied for any value of $k$, but a common choice of $k$ is to maximize the eigengap between successive eigenvalues $\lambda_{k} - \lambda_{k-1}$.

\bibliographystyle{_elsarticle-num-names}
\bibliography{__newrefs.bib}
\biboptions{sort&compress}
\end{document}